\documentclass[%
reprint,
superscriptaddress,
amsmath,
amssymb,
aps,
pra,
]{revtex4-2}

\usepackage[english]{babel}
\usepackage[usenames]{color}
\usepackage[document]{ragged2e}

\usepackage{amsfonts}
\usepackage{amsmath}
\usepackage{amssymb}
\usepackage{bm}
\usepackage{braket}
\usepackage{caption}
\usepackage{dcolumn}
\usepackage{enumitem}
\usepackage[final]{graphicx}
\usepackage{lineno}
\usepackage{mathtools}
\usepackage{orcidlink}
\usepackage{soul}
\usepackage{subcaption}
\usepackage{url}
\usepackage{verbatim}
\usepackage{xfrac}
\usepackage{xspace}
\usepackage{amsthm}
\usepackage{hhline}
\renewcommand{\arraystretch}{1.2}
\usepackage{algorithm}
\usepackage{algpseudocode}
\usepackage{hyperref}

\newtheorem{theorem}{Theorem}
\newtheorem{lemma}{Lemma}

\begin{document}

\title{From Hilbert's Tenth Problem to Quantum Speedup: \\
Explicit Oracles for Bounded Diophantine Systems}

\author{Gabriel Escrig}
\email{gescrig@ucm.es}
\affiliation{Departamento de Física Teórica, Universidad Complutense de Madrid.}

\author{M. A. Martín-Delgado}
\email{mardel@ucm.es}
\affiliation{Departamento de Física Teórica, Universidad Complutense de Madrid.}
\affiliation{CCS-Center for Computational Simulation, Universidad Politécnica de Madrid.}

\justifying

\begin{abstract}
Solving non-linear Diophantine systems lies at the mathematical core of integer optimization and cryptography. While the general unbounded problem is undecidable, even over bounded integer domains it remains classically intractable in the worst case. In this work, we introduce a fully reversible quantum algorithmic framework tailored to solve arbitrary polynomial Diophantine equations over bounded integer domains. The core of our approach is the explicit, gate-level synthesis of an evaluation oracle for amplitude amplification. By coherently evaluating polynomial constraints via in-place two's complement arithmetic and routing operations into a single recycled accumulator, this garbage-free strategy achieves a compact and scalable synthesis of the underlying non-linear arithmetic. Through analytical derivations and empirical circuit simulations, we prove that the overall spatial complexity is bounded by $q = \mathcal{O}((n + d^2)\log_2 N)$ logical qubits for $n$ variables, maximum degree $d$, and interval length $N$. The non-Clifford Toffoli depth is upper-bounded by $\mathcal{O}(q^2)$. This structural scaling exponent remains invariant to the variable count, modulated linearly only by the coefficients' Hamming weights. By moving beyond abstract black-box assumptions, this explicit architectural synthesis guarantees that the necessary quantum arithmetic acts as a bounded polynomial overhead. This ensures a quadratic speedup over classical exhaustive search, whether retrieving a unique assignment or dynamically enumerating an unknown number of solutions.
\end{abstract}

\maketitle
\justifying

\section{Introduction}\label{sec:intro}

\begin{figure*}[t]
  \centering
  \begin{subfigure}{0.49\textwidth}
    \centering
    \includegraphics[width=\linewidth]{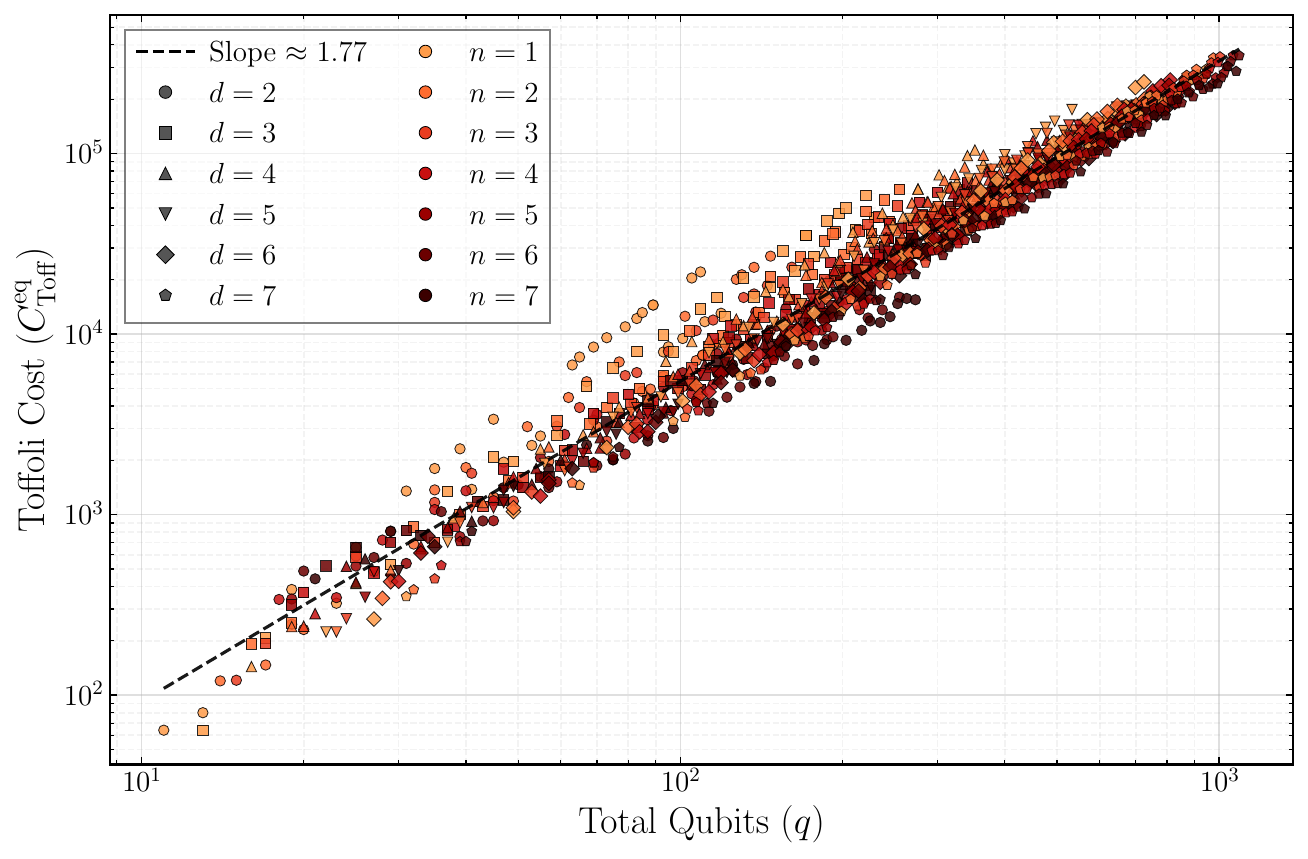}
    \label{fig:toffoli_scaling}
  \end{subfigure}\hfill
  \begin{subfigure}{0.49\textwidth}
    \centering
    \includegraphics[width=\linewidth]{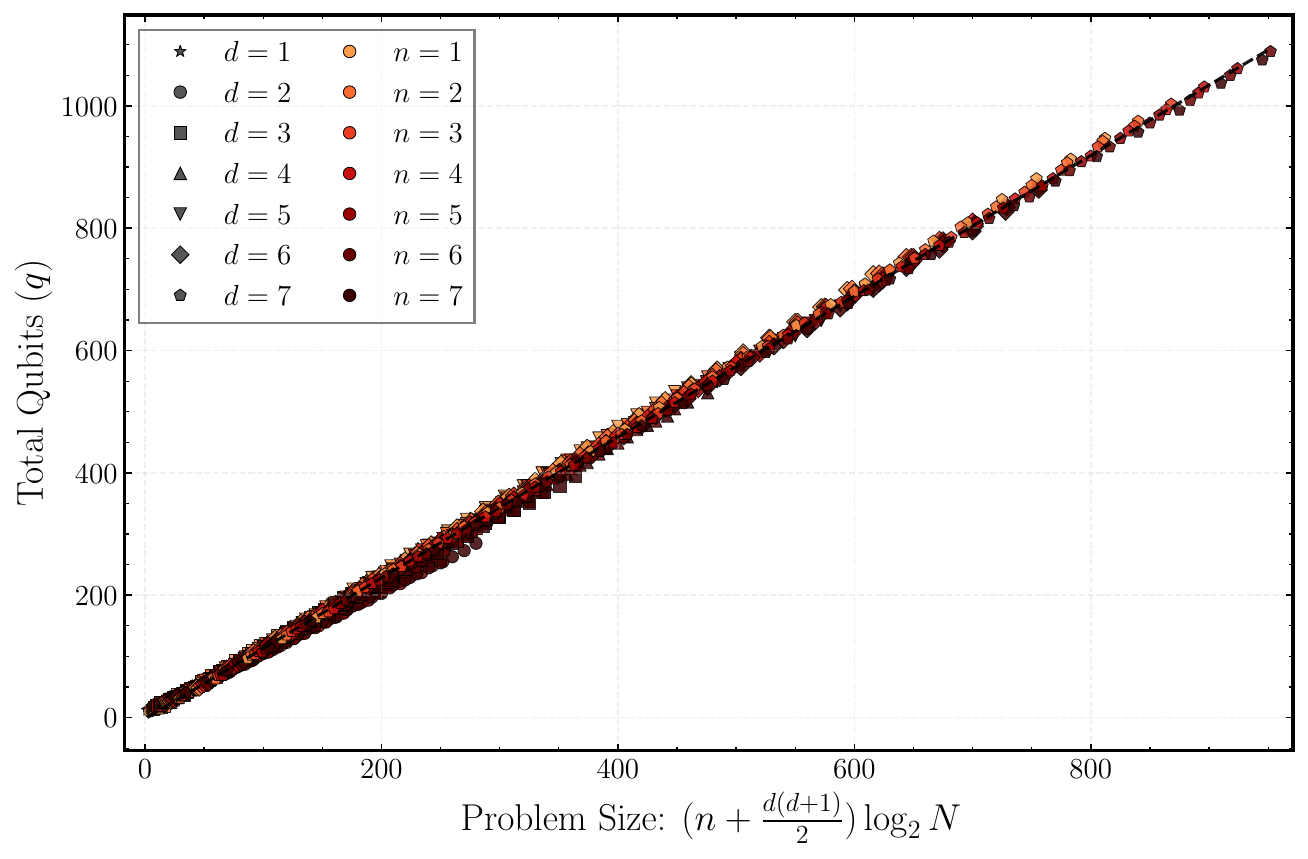}
    \label{fig:qubit_overhead}
  \end{subfigure}
  \captionsetup{font=small}
    \caption{\justifying 
    \textbf{Complexity analysis of the proposed quantum architecture.} Resource scaling is evaluated across 1500 Diophantine problem instances to establish empirical bounds. 
    \textbf{(Left)} Log-log scaling of the Toffoli gate count as a function of the total number of logical qubits $q$. Data points are generated by varying the number of variables $n \in [1, 7]$, the maximum polynomial degree $d \in [2, 7]$, and the search interval length $N$. The empirical scaling exponent of $1.77$ validates the sub-quadratic growth of the gate complexity in the pre-asymptotic regime.
    \textbf{(Right)} Total qubit requirement $q$ versus the theoretical problem size, $(n + \frac{1}{2}d(d+1)) \log_2 N$, for $n, d \in [1, 7]$. The results demonstrate a strictly linear spatial overhead, confirming the efficiency of the dynamic in-place arithmetic implementation.
  }
  \label{fig:resource_estimation_combined}
\end{figure*}

Diophantine systems of equations, in which all unknowns are required to take integer values, occupy a central role in number theory, discrete mathematics, and theoretical computer science. Their study dates back to classical problems posed by Diophantus and underpins modern developments in algebraic number theory and integer optimization~\cite{HardyWright, Cohen1993}. Despite their seemingly simple algebraic structure, these systems exhibit a level of computational complexity that is substantially higher than that of their continuous counterparts. While the solvability of a real linear system can be decided efficiently through rank conditions, the integer-constrained setting involves number-theoretic compatibility requirements and quickly becomes computationally intractable. Even in the linear case, determining feasibility is closely related to integer programming and is known to be NP-complete in bounded formulations~\cite{Karp1972, GareyJohnson1979}, while the general unbounded non-linear problem is formally undecidable, as demonstrated by Matiyasevich's resolution of Hilbert's tenth problem~\cite{MATIJASEVI2003}. Consequently, practical applications almost exclusively consider bounded domains, where the fundamental undecidability collapses into extreme classical intractability. It is precisely this bounded, decidable-but-NP-complete regime that defines the operational scope of the present work, where the quantum speedup we develop becomes meaningful.

This inherent computational hardness is not merely a theoretical curiosity; it constitutes a formidable bottleneck across diverse domains of applied mathematics, operations research, and information security. In the realm of operations research, many complex scheduling, resource allocation, and combinatorial optimization tasks are naturally formulated as Non-Linear Integer Programming problems, which fundamentally reduce to solving bounded Diophantine systems~\cite{Hemmecke2010}. Similarly, in computer science, automated theorem proving and formal software verification often rely on Satisfiability Modulo Theories solvers that must navigate non-linear arithmetic constraints over discrete integers~\cite{Barrett2018_SMT}. Furthermore, the hardness of discrete algebraic optimization forms the mathematical bedrock of cryptography. Beyond classical integer factorization, emerging Post-Quantum Cryptography paradigms---such as Lattice-based assumptions~\cite{10.1561/0400000074} and Multivariate Quadratic signature schemes~\cite{Ding}---rely fundamentally on the intractability of finding integer vectors that satisfy highly coupled polynomial equations. Advancing the algorithmic resolution of non-linear integer equations therefore has profound, cross-disciplinary implications.

In domains where such problems lack exploitable algebraic structure, classical analytical methods falter. The discrete structure of integer domains strips away the continuous symmetries that allow for efficient resolution via gradient methods or standard continuous relaxations. While specific linear subclasses can be approached via lattice reduction, determining the solvability of even simple bounded binary quadratic equations remains NP-complete~\cite{Manders1978}. In these unstructured regimes, heuristic reductions fail, rendering exhaustive enumeration over a defined, bounded search space $\mathcal{D}$ the only strategy guaranteed to find a solution or definitively prove its non-existence. Letting $|\mathcal{D}|$ be the size of this finite domain, classical deterministic algorithms are forced into an $\mathcal{O}(|\mathcal{D}|)$ brute-force verification. It is precisely in this context that quantum computation offers a definitive advantage. By mapping the discrete search space to a quantum superposition, Grover's algorithm~\cite{10.1145/237814.237866} allows for solution recovery with $\mathcal{O}(\sqrt{|\mathcal{D}|})$ oracle queries, providing a provable quadratic speedup over the classical baseline~\cite{PhysRevA.60.2746}. However, realizing this theoretical advantage in practice requires moving beyond abstract black-box metrics to develop concrete, resource-efficient quantum oracles.

Throughout this work, we restrict our attention to bounded Diophantine feasibility over finite integer intervals. The fundamental undecidability of the general problem remains; our contribution explicitly addresses the quantum-enhanced exhaustive search regime, where classical enumeration is the only viable strategy. Moving beyond abstract query complexity, we present an explicit, end-to-end construction of the Grover oracle tailored to these bounded polynomial systems. We detail the reversible arithmetic circuits required to evaluate constraints coherently and dynamically, without any reliance on technically prohibitive quantum-RAM assumptions. To overcome the massive spatial overhead of generic polynomial multipliers, we introduce a highly optimized \textit{shift-and-add} arithmetic engine that exploits zero-cost virtual rewiring for scalar multiplication. By nesting this fundamental engine within a novel reversible \textit{compute-utilize-uncompute} architecture, we successfully mitigate the severe circuit depth overhead typically associated with nested quantum arithmetic. We analytically demonstrate that the non-Clifford Toffoli gate complexity of the Grover operator is strictly bounded above by a worst-case theoretical limit of $\mathcal{O}(q^2)$, where $q$ is the total number of logical qubits. More importantly, our empirical evaluation reveals that the practical circuit depth exhibits a highly efficient sub-quadratic scaling well within this theoretical bound, as highlighted in Figure~\ref{fig:resource_estimation_combined}. Furthermore, the framework achieves a highly compact circuit width, bounding the total space complexity to $q = \mathcal{O}((n + d^2)\log_2 N)$ for a system of $n$ variables, maximum degree $d$, and search interval of width $N$.

The remainder of this article is organized as follows. Section \ref{sec:diophantine_systems} establishes the formal mathematical framework and complexity landscape of Diophantine systems. Section \ref{sec:overview} provides a high-level overview of the proposed quantum algorithmic approach, focusing on the mechanics of Grover's search applied to discrete integer spaces. Section \ref{sec:quantum_representation} introduces the core two's complement encoding scheme and develops the generalized compute-utilize-uncompute methodology and shift-and-add arithmetic for coherently constructing polynomials of arbitrary degree $d$. Building on this foundation, Section \ref{oracle_implementation} details the explicit circuit implementation of the evaluation oracle, establishing the formal proofs for our claimed complexity bounds. Section \ref{sec:results} presents comprehensive numerical simulations, featuring resource scaling analysis and exact precision plots that empirically validate the theoretical efficiency and correctness of the Grover operator. Finally, Section \ref{sec:conclusions} summarizes our findings and outlines future research directions.

\section{Diophantine Systems of Equations} \label{sec:diophantine_systems}
\begin{figure*}[t]
    \centering 
    \includegraphics[width=0.6\textwidth]{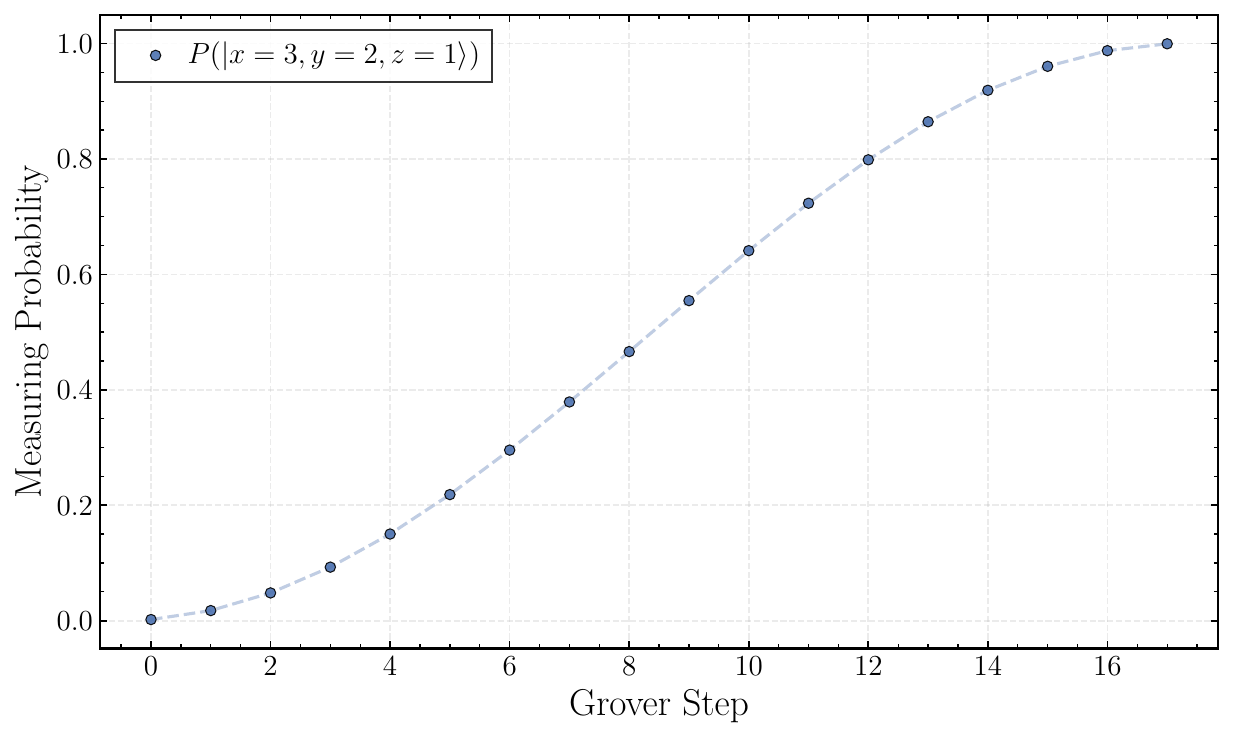}
    \captionsetup{justification=justified, font=small}
    \caption{\justifying \textbf{Amplitude amplification dynamics for a strongly coupled multivariate quadratic Diophantine system.} The plot displays the probability of measuring the unique target state $P(\ket{x=3, y=2, z=1})$ as a function of the number of Grover iterations. The underlying quantum oracle evaluates the explicit non-linear system: $3x^2+2y^2+5z^2 = 40$, $2xy-4yz+3xz = 13$, and $-x^2+5y-7z = -6$. Each of the $n=3$ variables is discretized using a 3-qubit register, yielding a global search space of size $|\mathcal{D}|=2^9=512$. The coherent search successfully navigates the complex arithmetic landscape, with the success probability peaking at $99.9\%$ exactly at the theoretically predicted optimal step $t_{\mathrm{opt}} \approx \lfloor \frac{\pi}{4}\sqrt{|\mathcal{D}|} \rfloor = 17$, thereby validating the phase accuracy of the fully reversible integer arithmetic implementation.
    }
    \label{fig:prob_vs_step}
\end{figure*}

A Diophantine system is a collection of polynomial equations with integer coefficients in which all unknowns are required to take integer values. Formally, for integers $m, n \geq 1$, one considers polynomials
\begin{equation}\label{eq:diophantine_general}
    f_{j} \in \mathbb{Z}[x_{1}, \dots, x_{n}], \qquad j = 1,\dots,m,
\end{equation}
and seeks an integer vector
\begin{equation}
    \mathbf{x} = (x_{1},\dots,x_{n}) \in \mathbb{Z}^{n},
\end{equation}
satisfying the system
\begin{equation}\label{eq:syst_dioph}
    f_{1}(\mathbf{x}) = f_{2}(\mathbf{x}) = \cdots = f_{m}(\mathbf{x}) = 0.
\end{equation}

Eq.~\eqref{eq:syst_dioph} defines the general Diophantine feasibility problem. Its complexity landscape is structured by the choice of solution domain. Hilbert's tenth problem in its general unbounded form---deciding whether Eq.~\eqref{eq:syst_dioph} admits a solution in $\mathbb{Z}^n$---is undecidable~\cite{MATIJASEVI2003}. When the search is restricted to a bounded integer domain, the problem becomes decidable but remains NP-complete in several natural formulations: 0-1 integer programming~\cite{Karp1972}, bounded quadratic Diophantine equations~\cite{Manders1978}, and bounded integer programming feasibility~\cite{BoroshTreybig}.

In this work, we consider the explicit bounded search space $\mathcal{D} = [-N/2, N/2-1]^n$, with $N$ given as part of the input. In this regime, the problem is decidable by exhaustive enumeration over $\mathcal{D}$, and the classical worst-case cost is $\Theta(N^n \cdot C_{\text{eval}})$, where $C_{\text{eval}}$ is the bit-complexity of evaluating the polynomial system on a single candidate. This deterministic baseline is precisely the target against which our quantum framework, built on coherent amplitude amplification, provides a quadratic reduction to $\mathcal{O}(\sqrt{N^n/M})$ oracle queries for $M \ge 1$ valid assignments.

This sharp transition in complexity fundamentally limits the applicability of standard continuous quantum algorithms. In the continuous domain ($\mathbb{R}^{n}$), algorithms such as the HHL protocol~\cite{PhysRevLett.103.150502} offer potential exponential quantum speedups for solving linear systems, provided highly favorable conditions are met---specifically, for sparse, well-conditioned matrices alongside efficient state preparation. However, the very nature of such protocols constitutes a structural bottleneck in the Diophantine setting. Amplitude-encoded algorithms prepare a quantum state where the probability amplitudes themselves represent the continuous floating-point solution. They provide no algebraic mechanism to enforce strict integrality constraints, nor can they distinguish between valid integer roots and non-integer approximations.

Consequently, the discrete arithmetic hardness persists. The inability of standard quantum linear algebra to accommodate the discrete nature of $\mathbb{Z}^n$ dictates a mandatory paradigm shift. To coherently resolve bounded Diophantine systems---whether linear or highly non-linear---one must abandon continuous amplitude encoding in favor of exact digital quantum arithmetic. This structural necessity motivates the approach explored in this work: explicitly embedding the bounded integer search space into a quantum superposition and utilizing a fully reversible, exact algebraic oracle to systematically isolate valid integer assignments.

\begin{figure*}[t] 
    \centering \includegraphics[width=1\textwidth]{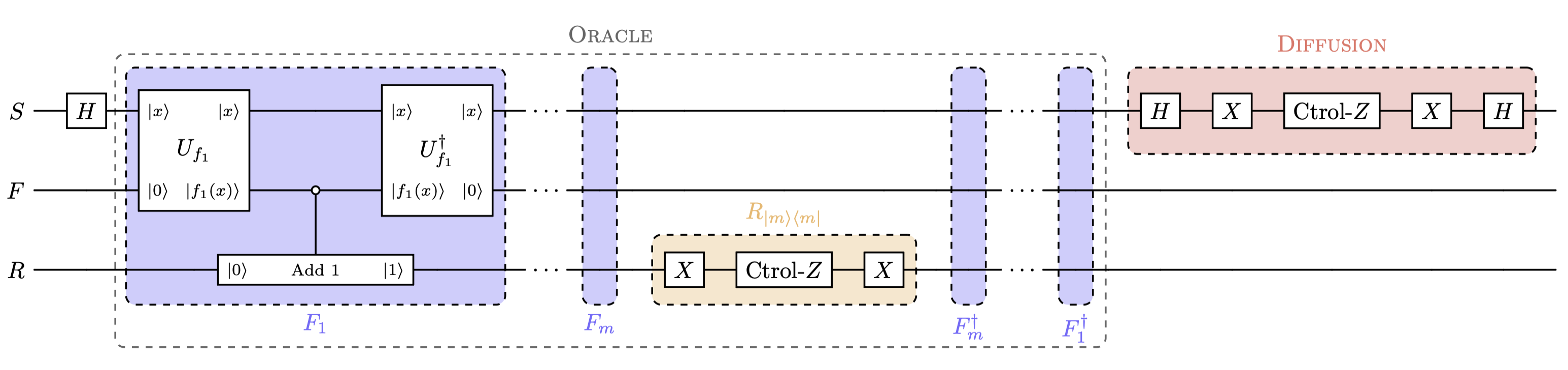}
    \captionsetup{font=small}
    \caption{\justifying \textbf{Schematic representation of a single Grover iteration.} The circuit explicitly constructs the composite operator $G=DO_f$, highlighting the sequential application of the Diophantine oracle (left block), which coherently evaluates the system of $m$ polynomial equations to mark valid solution states, followed by the diffusion operator (right block), which amplifies their amplitudes.}
    \label{fig:Grover_iteration}
\end{figure*}

\section{Overview of the Quantum Search Strategy} \label{sec:overview}

Our objective is to determine the integer solutions of a Diophantine system within a bounded domain. Since no general deterministic method exists for solving arbitrary Diophantine equations, exhaustive search over the discrete space remains the most direct---and in unstructured cases, the only viable---strategy. Grover's algorithm~\cite{10.1145/237814.237866, RevModPhys.74.347} provides a provable quadratic speedup over classical enumeration, rendering it the natural quantum analogue for solving bounded Diophantine systems.

Grover's search operates by iteratively amplifying the probability amplitudes of ``marked'' basis states, which in this context correspond to integer assignments $\mathbf{x}$ that satisfy the Diophantine system. This amplification is driven by two alternating unitaries: the oracle, which applies a selective phase flip to feasible solutions, and the diffusion operator, which performs an inversion about the mean of the probability amplitudes. After a specific number of iterations proportional to $\sqrt{|\mathcal{D}|/M}$ (where $|\mathcal{D}|$ is the total size of the discrete search space and $M$ is the number of valid solutions), the probability of measuring a satisfying assignment approaches unity, as explicitly demonstrated for a highly coupled non-linear system in Figure~\ref{fig:prob_vs_step}. It is worth noting that the search operator used in this work corresponds to the canonical Grover choice, which may be regarded as a distinguished member of a broader family of Grover-type quantum search algorithms. Generalized kernels of this kind retain the characteristic quadratic speedup, provided the underlying phase relations are chosen consistently~\cite{PhysRevA.62.062303}.

To embed this mathematical problem into quantum hardware without collapsing the superposition, our architecture dynamically allocates three primary quantum registers. First, the System register ($S$) stores the integer assignments $\mathbf{x}$ in superposition. Second, the Function register ($F$) acts as a reversible arithmetic accumulator to evaluate the polynomials $f_j(\mathbf{x})$ on-the-fly. Finally, the Result counter ($R$) tracks the exact number of equations satisfied by each quantum state. The precise orchestration of these registers ensures that the oracle operates strictly unitarily.

The complete algorithmic pipeline for solving Diophantine systems is depicted schematically in Figure~\ref{fig:Grover_iteration}, and proceeds as follows:

\begin{enumerate}[label=\arabic*), leftmargin=2.2em]

  \item[\bf (1)] \emph{Initialization.}  
        Construct a uniform superposition over the total bounded integer domain $\mathcal{D}$. Assuming each of the $n$ variables is encoded as a $w$-bit two's complement integer, the search space is $\mathcal{D} = [-2^{w-1}, 2^{w-1}-1]^n$, with total cardinality $|\mathcal{D}| = 2^{nw}$:
        \begin{equation}
            \ket{\psi^{(0)}} 
            = \frac{1}{\sqrt{|\mathcal{D}|}} \sum_{\mathbf{x}\in\mathcal{D}} \ket{\mathbf{x}}_{S}.
        \end{equation}
        Here, $S$ denotes the system register storing the encoded integer vector $\mathbf{x}$.

  \item[\bf (2)] \emph{Oracle application.}  
        Apply the Diophantine oracle $O_f$, which applies a conditional negative phase exclusively to the solution states:
        \begin{equation}
            O_f \ket{\mathbf{x}}_S 
            = (-1)^{\mathrm{Sol}(\mathbf{x})}\ket{\mathbf{x}}_S,
        \end{equation}
        where $\mathrm{Sol}(\mathbf{x}) = 1$ if $\mathbf{x}$ satisfies all polynomial equations (\ref{eq:syst_dioph}), and $0$ otherwise. The physical synthesis of $O_f$ constitutes the primary technical challenge of this work and is detailed in Section~\ref{sec:quantum_representation}.

  \item[\bf (3)] \emph{Diffusion step.}  
        Apply the Grover diffusion operator
        \begin{equation}
            D = 2\ket{\psi^{(0)}}\!\bra{\psi^{(0)}} - I,
        \end{equation}
        which performs an inversion about the mean of the probability amplitudes.

  \item[\bf (4)] \emph{Grover iteration.}  
        Perform $t$ iterations of the composite search operator $G = D\,O_f$, producing the state:
        \begin{equation}
            \ket{\psi^{(t)}} = G^{t}\ket{\psi^{(0)}}.
        \end{equation}
        If the number of feasible solutions $M$ is known \textit{a priori}, the optimal iteration count is:
        \begin{equation}
            t_{\mathrm{opt}} \approx \left\lfloor\frac{\pi}{4}\sqrt{\frac{|\mathcal{D}|}{M}}\right\rfloor.
        \end{equation}
        Remarkably, even when $M$ is strictly unknown---as is typical in most practical scenarios---the algorithmic framework remains fully functional and preserves the quantum speedup. The specific strategies employed to iteratively retrieve, enumerate, or count these solutions without prior calibration are detailed in Section~\ref{sec:quantum_counting}.

  \item[\bf (5)] \emph{Measurement.}  
        Measure the system register $S$ in the computational basis. With high probability, the observed integer vector $\mathbf{x}$ constitutes a valid solution.
\end{enumerate}

This framework provides a fully quantum methodology for Diophantine resolution. However, its practical viability hinges entirely on the internal architecture of the oracle $O_f$. Unlike approaches that rely on quantum-RAM to query pre-computed look-up tables or extensive classical pre-processing, we propose a fully quantum, on-the-fly arithmetic construction. Evaluating arbitrary polynomial equations coherently requires meticulous reversible logic. The explicit synthesis of these arithmetic circuits---including the shift-and-add arithmetic and the sequential compute-utilize-uncompute methodology---is developed in detail in Section~\ref{sec:quantum_representation}.

\begin{figure*}[t] 
    \centering 
    \includegraphics[width=1\textwidth]{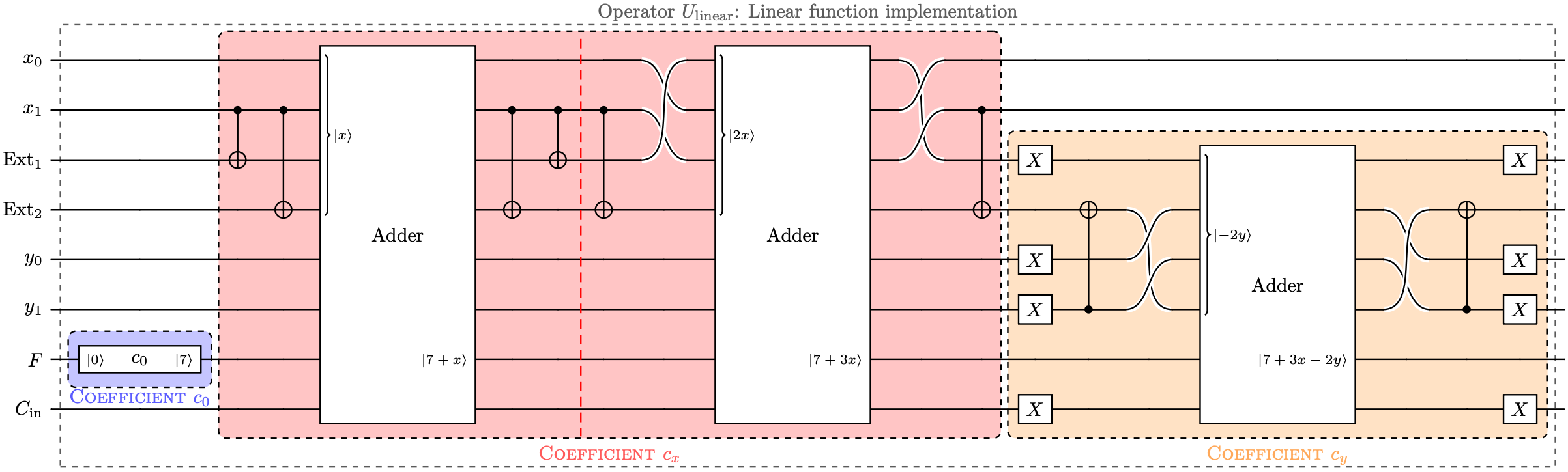}
    \captionsetup{font=small}
    \caption{\justifying \textbf{Circuit schematic of the logical rewiring technique.} The implementation is demonstrated for the evaluation of the linear function $f(x,y) = 3x - 2y+7$. By exploiting cost-free logical bit-shifts (rewiring), the need for full quantum multipliers is eliminated. Since the coefficient of $x$ ($c_x = 3 \equiv 11_2$) has a Hamming weight of $w_H=2$, its scalar multiplication is strictly decomposed into two sequential quantum additions. Conversely, to implement the negative coefficient $c_y = -2$, the circuit utilizes two's complement arithmetic: the input qubits of the $y$ register are inverted via $X$ gates and the carry-in bit ($C_{\mathrm{in}}$) is set to $|1\rangle$. The scalar factor of 2 is concurrently resolved through a 1-bit logical shift prior to a single quantum addition.}
    \label{fig:re-wiring}
\end{figure*}

\section{Quantum Representation of Diophantine Systems} 
\label{sec:quantum_representation}

To explicitly construct the evaluation oracle $O_f$, we must map the discrete domain of integer variables to the state space of a quantum register. By definition, Diophantine systems restrict every equation to polynomial forms; evaluating them coherently therefore reduces to computing weighted sums and multiplications of integer variables. Since bounded multiplication can be synthesized from nested additions, the quantum adder emerges as the fundamental arithmetic primitive~\cite{PhysRevA.54.147, Wang2025}. Building on these reversible logic blocks, this section details the fully quantum arithmetic architecture required to evaluate arbitrary polynomials coherently in superposition.

\subsection{Two's Complement Encoding}
\label{subsec:twos_complement}

Since our construction operates over integer-valued variables, it is necessary to encode both positive and negative numbers within a reversible quantum representation. To this end, we adopt the standard two's complement binary encoding, which provides a uniform and arithmetic-friendly treatment of signed integers in quantum circuits. Let $x \in \mathcal{D} \subset \mathbb{Z}$ be an integer variable restricted to a bounded domain. We encode $x$ into a quantum register of $w$ qubits, where a computational basis state $\ket{x} = \ket{b_{w-1}b_{w-2}\dots b_0}$, with $b_i \in \{0, 1\}$, represents the integer value:
\begin{equation}
    x = -b_{w-1}2^{w-1} + \sum_{i=0}^{w-2} b_i 2^i.
\end{equation}
Here, the most significant bit (MSB), $b_{w-1}$, serves as the sign bit. The state space of the $w$-qubit register spans the symmetric interval $\mathcal{D}_k = [-2^{w-1}, 2^{w-1}-1]$. For a Diophantine system comprising $n$ independent variables $\mathbf{x} = (x_1, \dots, x_n)$, the joint state is prepared in a composite Hilbert space $\mathcal{H}_{S} = \bigotimes_{j=1}^n \mathcal{H}_{x_j}$, defined by the tensor product of the individual variable registers.

A major algorithmic advantage of this representation is that quantum subtraction can be seamlessly reduced to addition through the standard arithmetic identity:
\begin{equation} \label{eq:twos_comp_subtraction}
    a - b = a + \overline{b} + 1,
\end{equation}
where $\overline{b}$ denotes the bitwise complement of $b$. In our quantum circuit design, this operation is realized by applying $X$ gates to each qubit of the register encoding $b$ prior to the addition. Crucially, the mandatory increment by one is absorbed at zero additional gate depth by simply initializing the controlled carry-in ($C_{\mathrm{in}}$) qubit of the quantum adder to the $\ket{1}$ state. This elegant reduction ensures that both addition and subtraction are handled uniformly within the exact same reversible framework, avoiding the need for dedicated subtraction primitives.

\subsection{Reversible Evaluation of Linear Diophantine Equations}
\label{subsec:linear_arithmetic}

Linear Diophantine equations constitute the simplest and most fundamental class of integer polynomial equations. As such, they provide a natural starting point for the development of coherent quantum evaluation techniques. We begin by constructing an explicit, fully reversible quantum circuit that evaluates a general linear integer function of the form:
\begin{equation}\label{eq:linear}
    f(x_1, x_2, \ldots, x_n) = c_0 + \sum_{i=1}^n c_i x_i.
\end{equation}

To evaluate the scalar multiplication of a variable $x_i$ by a classical integer coefficient $c_i$, we implement a highly optimized \textit{shift-and-add} architecture. By decomposing the absolute value of the coefficient into its binary expansion, $|c_i| = \sum_j p_j 2^j$ (with $p_j \in \{0,1\}$), the operation is reduced to conditionally adding $x_i$ shifted by $j$ positions for every $j$-th bit of $|c_i|$ that is $1$. This strategy bypasses the severe $\mathcal{O}(|c_i|)$ arithmetic overhead demanded by naive repeated addition approaches, securing an exponentially more efficient scaling. 

Crucially, the bit-shifts corresponding to $2^j$ do not require explicit quantum shift gates. They are implemented strictly through \textit{virtual rewiring}. As explicitly detailed in Figure~\ref{fig:re-wiring}, by dynamically routing the wires of the variable register $x_i$ into shifted input terminals of the addition network, we effectively multiply by powers of two at zero gate cost.

We adopt the Cuccaro--Draper--Kutin--Moulton (CDKM) ripple-carry adder~\cite{Cuccaro:2004xxx} as our core arithmetic primitive owing to its qubit-optimal circuit width. To maintain a strictly in-place evaluation, all partial products are accumulated directly into a single target evaluation register, $\ket{0}_F$. Since all inputs to the CDKM adder must have identical bit widths, a dynamic sign-extension is applied whenever $x_i$ is logically shifted. This is achieved by padding the least significant bits with $\ket{0}$ ancillas and copying the sign bit (the MSB of $x_i$) into the upper extension qubits using CNOT gates, ensuring that the two's complement representation remains globally consistent.

To handle negative coefficients ($c_i < 0$), we seamlessly leverage the subtraction identity of two's complement arithmetic previously introduced in Eq.~\eqref{eq:twos_comp_subtraction}. When $c_i$ is negative, we pre-invert the logically shifted $x_i$ register using $X$ gates. Furthermore, because ripple-carry adder architectures such as the CDKM inherently feature a carry-in ($C_{\mathrm{in}}$) register (see Figure~\ref{fig:re-wiring}), we assert a logical $\ket{1}$ on this $C_{\mathrm{in}}$ qubit. This efficiently computes the exact two's complement of the shifted variable on-the-fly, reducing subtraction to addition without introducing new arithmetic primitives. All ancillary entanglements (such as the MSB sign-extensions) and temporary $X$ inversions are strictly uncomputed locally after each addition step. This local uncomputation ensures that all auxiliary qubits are deterministically restored to their initial unentangled state, preventing the accumulation of residual garbage entanglement that would otherwise destroy the critical phase coherence required by the subsequent Grover diffusion operator.

The linear evaluation circuit is initialized by encoding the constant term $c_0$ into the function register $F$ using a layer of $X$ gates. The procedure then proceeds sequentially for each variable and its corresponding active bits, yielding the global unitary transformation:
\begin{equation} \label{function_operators}
        U_{\mathrm{linear}} \ket{\mathbf{x}}_{S}\ket{0}_{F} 
        = \ket{\mathbf{x}}_{S}\ket{f(\mathbf{x})}_{F}.
\end{equation}

By replacing generic, resource-heavy quantum multipliers with coefficient-specific shift-and-add routines and zero-cost virtual rewiring, this architecture minimizes both Toffoli gate complexity and ancillary qubit overhead. Formally, the number of required quantum additions to evaluate $c_i x_i$ is exactly given by the Hamming weight of the coefficient's binary expansion, denoted as $w_H(|c_i|)$, which corresponds to the number of non-zero bits. Since $w_H(|c_i|) \le \lfloor\log_2|c_i|\rfloor + 1$, the operation count is exponentially reduced compared to unary accumulation. Furthermore, a single in-place addition into the function register requires a Toffoli count that scales linearly with the active register width~\cite{Wang2025}. Because this active computational width is strictly bounded by the total number of logical qubits $q$, the worst-case non-Clifford cost to evaluate a single linear term is bounded by $\mathcal{O}(q \cdot w_H(|c_i|))$. Extending this to the complete linear polynomial, the cumulative non-Clifford gate complexity is strictly bounded by $\mathcal{O}\left(q \sum_{i=1}^n w_H(|c_i|)\right)$. Consequently, the global linear evaluation circuit exhibits a gate complexity that is strictly linear with respect to both the overall logical space $q$ and the collective Hamming weights of the coefficients. This theoretical bound is definitively validated through exact empirical circuit synthesis in Section~\ref{sec:results}.

\subsection{Reversible Evaluation of Quadratic Diophantine Equations}
\label{subsec:quadratic_arithmetic}

\begin{figure*}[t] 
    \centering 
    \includegraphics[width=1\textwidth]{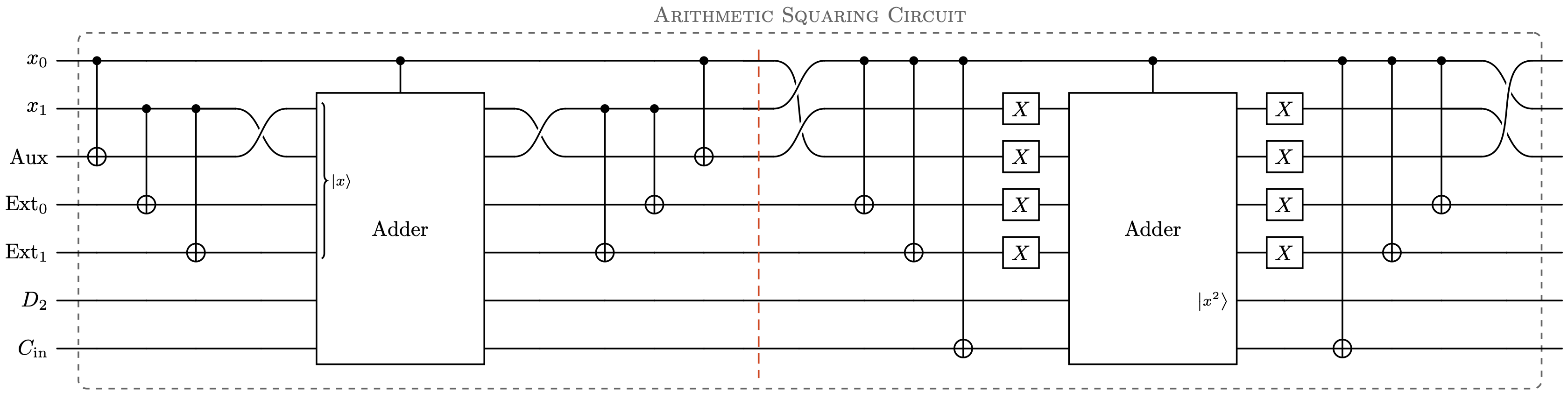}
    \captionsetup{font=small}
    \caption{\justifying \textbf{Signed quantum squaring circuit.} This schematic illustrates the in-place evaluation of a quadratic term. To prevent algorithmic self-reference conflicts, an ancillary qubit is temporarily entangled with the control bit via a CNOT gate, mediating its input to the adder. The final step, controlled by the sign bit (MSB), triggers the two's complement subtraction logic ($X$ inversions and $C_{\mathrm{in}}=\ket{1}$) to properly weight the negative MSB component.}
    \label{fig:squaring}
\end{figure*}

Having established an optimized reversible procedure for evaluating linear functions, we naturally extend the construction to quadratic equations. Quadratic forms introduce multiplicative interactions---specifically, self-squaring and cross-terms---while still admitting a structured polynomial decomposition:
\begin{equation}
    f(\mathbf{x}) = c_0 + \sum_{i=1}^n c_i x_i + \sum_{i=1}^n a_i x_i^2 + \sum_{i<j} b_{ij} x_i x_j.
\end{equation}

A conventional approach to evaluating non-linear terms involves generic reversible multipliers, which require storing intermediate full-width products (e.g., $a_i x_i^2$) in auxiliary data registers before further accumulation. This method is highly resource-inefficient. Instead, we propose a direct \textit{in-place accumulation} strategy. Because quadratic terms involve exactly two variable operands, we can generalize the linear shift-and-add algorithm into self-controlled and cross-controlled additions, allowing the entire quadratic form to be evaluated directly into the main function register $\ket{0}_F$ without relying on any intermediate full-width accumulators.

\subsubsection{In-Place Squaring Operator for Signed Integers}

To evaluate pure quadratic terms of the form $a_i x_i^2$, we employ a dedicated Squaring Operator that executes strictly in-place. Let the variable $x_i$ be encoded in a $w$-bit register, $\ket{x_{i,w-1} \dots x_{i,0}}$. The circuit essentially concatenates a sequence of $w$ controlled adders. At the $v$-th addition stage ($v \in \{0, \dots, w-1\}$), the register $x_i$ is logically shifted by $v$ positions via virtual rewiring and added to the accumulator $F$, with the addition strictly controlled by its own qubit $x_{i,v}$. This elegantly computes the partial product $x_{i,v} \cdot (x_i \cdot 2^v)$. 

However, implementing this naively violates a fundamental constraint of quantum circuit design: a single physical qubit cannot simultaneously act as the control of a unitary operation and as a data operand within the target register of that same unitary. To break this self-referential loop, a single ancillary qubit is temporarily entangled with the control qubit $x_{i,v}$ via a CNOT gate. The auxiliary qubit is then fed into the adder's input array, preserving the unitary nature of the operation while fully isolating the control line (see Figure~\ref{fig:squaring}).

Handling two's-complement arithmetic for non-linear terms presents an additional critical challenge. Because the MSB of a $w$-bit two's-complement integer carries a negative algebraic weight ($-x_{i,w-1} 2^{w-1}$), the final iteration of the squaring loop---when the control is the MSB---mathematically mandates a subtraction rather than an addition. Furthermore, if the global coefficient $a_i$ is itself negative, the entire addition/subtraction logic must be inverted. We resolve this elegantly through a unified Boolean parity approach. For any shift dictated by the coefficient and the variable, the decision to subtract the shifted register is governed by the logical XOR condition:
\begin{equation}
    s = x_{i, w-1} \oplus s_{a_i},
\end{equation}
where $x_{i, w-1}$ represents the MSB of the quantum variable and $s_{a_i} \in \{0,1\}$ is the classical sign bit of the coefficient $a_i$ ($1$ if negative, $0$ otherwise). If $s = 1$, the input array is bitwise inverted with $X$ gates and the carry-in qubit is activated ($C_{\mathrm{in}}=\ket{1}$), seamlessly transforming the controlled addition into a controlled subtraction. Dynamic sign-extension is applied exclusively to the upper padding qubits, preventing the corruption of the shifted lower-order bits.

\begin{table*}[t]
\centering
\renewcommand{\arraystretch}{1.3}
\setlength{\tabcolsep}{10pt}
\small
\begin{tabular}{|c||l|c|}
\hline
\textbf{Register} & \textbf{Description} & \textbf{Qubits} \\
\hline \hline
$S$   & System register encoding the $n$ variables $\mathbf{x} \in \mathcal{D}$ & $n \lceil \log_2 N \rceil$ \\[2pt]
$F$   & Main accumulator for the polynomial evaluation of $\ket{\mathbf{x}}_S$ & $\max_{1 \le j \le m} \left\lceil \log_2 \left(\sum_{\alpha} |c_{j,\alpha}| N^{|\alpha|}\right) \right\rceil$ \\[2pt]
$D_\ell$ & Intermediate monomial registers for degrees $\ell \in \{2, \ldots, d-1\}$ & $\sum_{\ell=2}^{d-1} \ell \lceil \log_2 N \rceil$ \\[2pt]
$R$   & Equations counter register & $\lceil\log_2{(2m+1)}\rceil$ \\[2pt]
\hline
\end{tabular}
\caption{\justifying Summary of the quantum registers used in the generalized algorithm. Each variable $x_i$ is encoded in two's complement to search the interval $x_i \in [-N/2, N/2 -1]$ for $i = 1, \ldots, n$. Because $F$ is uncomputed and dynamically recycled, it is sized to accommodate the maximum possible polynomial value across all $m$ equations (using multi-index notation $\alpha$). The intermediate registers $D_\ell$ follow an arithmetic progression to store higher-order monomials during the \emph{Compute} step. A constant $\mathcal{O}(1)$ number of ancillas are omitted for brevity.}
\label{tab:registers}
\end{table*}

\subsubsection{Cross-Term Evaluation}

The evaluation of cross terms $b_{ij} x_i x_j$ follows naturally from the squaring architecture. Instead of a variable controlling its own shifted addition, the qubits of $x_i$ act as controls for the logically shifted additions of the disjoint register $x_j$. Specifically, for every $1$-bit in the binary expansion of $|b_{ij}|$ at position $p$, and for every qubit $v \in \{0, \dots, w-1\}$ in $x_i$, the register $x_j$ is shifted by $p+v$ positions. The addition is controlled by $x_{i,v}$, effectively accumulating $x_{i,v} \cdot (x_j \cdot 2^{p+v})$ directly into $F$. 

Because the control ($x_i$) and the target data ($x_j$) reside in strictly disjoint quantum registers, the ancillary qubit required in the squaring operator is entirely obviated, further streamlining the circuit. The exact same XOR-based two's complement correction logic ($s = x_{i, w-1} \oplus s_{b_{ij}}$) dictates the subtraction mechanism.

From a resource scaling perspective, compiling the quadratic terms $x_i^2$ or $x_i x_j$ with a generic coefficient $c$ requires a nested shift-and-add approach. Multiplying two $w$-bit variables inherently necessitates $w$ controlled additions for each active bit of the scalar coefficient. Because a single controlled in-place addition requires a Toffoli depth that scales linearly with the active register width, the local non-Clifford gate complexity to evaluate a single quadratic term scales as $\mathcal{O}(w \cdot q \cdot w_H(|c|))$, where $q$ serves as the global upper bound for the function register width. Enclosing this strictly within the total problem size $q$, the analytical worst-case complexity scales as $\mathcal{O}(q^2 \cdot w_H(|c|))$. Extending this to a general polynomial comprising $T$ quadratic terms, the cumulative non-Clifford cost scales proportionately as $\mathcal{O}\left(q^2 \sum_{t=1}^T w_H(|c_t|)\right)$. Crucially, because this $\mathcal{O}(q^2)$ non-linear cost asymptotically eclipses the strictly linear $\mathcal{O}(q)$ overhead of the terms evaluated in Section~\ref{subsec:linear_arithmetic}, the total gate complexity of a full quadratic Diophantine system is entirely dominated by this quadratic bound.

It is paramount to emphasize that this $\mathcal{O}(q^2)$ boundary represents an absolute theoretical maximum. In practice, the actual number of participating qubits is strictly smaller than $q$ during early accumulation stages. Furthermore, because the arithmetic operations are quantumly controlled, zero-valued control states dynamically bypass the addition logic without incurring full arithmetic depth. This combination of bounded active widths and dynamic zero-skipping explains the highly efficient, sub-quadratic empirical scaling initially introduced in Figure~\ref{fig:resource_estimation_combined} and comprehensively validated in Section~\ref{sec:results}. Ultimately, by completely eliminating generic full-width multipliers and temporary accumulators, this architecture achieves a remarkable level of qubit efficiency, relying solely on the main $F$ register and a single ancilla qubit.

\subsection{Reversible Evaluation of General Polynomial Diophantine Equations}

The optimized in-place accumulation strategy developed for linear and quadratic equations forms the fundamental arithmetic engine for evaluating polynomials of arbitrary degree $d \ge 3$. The principal challenge in scaling to higher degrees lies in the generation of higher-order monomials (e.g., $x_i^3$ or $x_i^2 x_j$). Unlike quadratic terms, which can be accumulated directly into the main function register $F$ using a single ancilla qubit (as established in Figure~\ref{fig:squaring}), evaluating degree-$d$ terms requires sequential multiplications. Attempting to compute these sequentially into $F$ without intermediate storage breaks fundamental reversibility constraints, while relying on generic full-width sequential multipliers wastes considerable gate depth.

To resolve this bottleneck, we introduce a structured \textit{Compute-Utilize-Uncompute} architecture inspired by Bennett's reversible computing paradigm~\cite{10.1147/rd.176.0525}. As detailed in Table~\ref{tab:registers}, rather than using full-width auxiliary accumulators, we provision a hierarchical set of intermediate, tightly bounded quantum registers, denoted as $D_\ell$ for $\ell \in \{2, \ldots, d-1\}$. Each register $D_\ell$ is strictly sized to $\ell \cdot \log_2N$ qubits, which is the exact precision required to hold a pure monomial of degree $\ell$ without arithmetic overflow. This strategy severs the spatial complexity's dependence on the number of variables $n$, limiting the auxiliary circuit width to $\mathcal{O}(d^2 \cdot \log_2N)$ (see Theorem~\ref{theor:complexity_total} for the formal derivation).

The evaluation of a general higher-order term, such as $c \cdot x_1 x_2 \cdots x_d$, proceeds systematically through three distinct quantum steps:

\begin{enumerate}[label=\arabic*), leftmargin=2.2em]

  \item[\bf (1)] \emph{Compute Step (Monomial Generation).}  
        The pure underlying monomial (ignoring the global coefficient $c$) is sequentially built up in the intermediate registers by repeatedly invoking the quadratic multiplication engine detailed in Figure~\ref{fig:squaring}. Starting from the base variables in $S$, local shift-and-add operations construct $x_1 x_2$ into $D_2$. Subsequently, $x_3$ acts as a control to shift-and-add $D_2$ into $D_3$. This localized cascade continues until the degree $(d-1)$ monomial is successfully prepared in register $D_{d-1}$. 
        
  \item[\bf (2)] \emph{Utilize Step (In-Place Accumulation).}  
        With the intermediate monomial prepared, the sub-register encoding the final variable $x_d$ acts as the control, while $D_{d-1}$ acts as the target register. We perform a reversible, controlled multiplication between the quantum state of the variable $\ket{x_d}$ and the intermediate state $\ket{ x_1 x_2\cdots x_{d-1}}_{D_{d-1}}$, dynamically accumulating the resulting product directly into the main function register $F$. Crucially, it is solely during this final accumulation step that the absolute global coefficient $|c|$ determines the exact number of arithmetic shifts---strictly dictated by its Hamming weight---while the XOR-based two's complement logic dynamically governs the sign of the addition.
            
  \item[\bf (3)] \emph{Uncompute Step (Garbage Uncomputation).}  
        To preserve strict reversibility and free the intermediate registers for the evaluation of subsequent polynomial terms, the \textit{Compute Step} is perfectly inverted. By applying the inverse local adders in reverse order, all intermediate registers from $D_{d-1}$ down to $D_2$ are deterministically returned to the computational zero state $\ket{0\dots 0}_{D_\ell}$. 
\end{enumerate}

From a resource scaling perspective, this architecture isolates the computational burden of the coefficient $c$ exclusively to the Utilize Step. The Compute and Uncompute steps are executed using purely unweighted ($c=1$) shift-and-add operations to build and dismantle the raw monomial. Because this process decomposes any specific degree-$d_t$ monomial into a cascade of pairwise operations, generating its intermediate states requires $\mathcal{O}(d_t)$ sequential multiplications. Crucially, the non-Clifford gate complexity of each intermediate multiplication is bounded by the active qubit width of the target register, which spans exactly $\ell \log_2 N$ qubits during computational step $\ell$.

While the cumulative gate count naturally grows with the polynomial degree, this sequential circuit overhead is structurally decoupled from the coefficient's Hamming weight, $w_H(|c_t|)$. To prepare each specific monomial, the circuit executes a sequence of multiplications, each incurring a gate cost that scales quadratically with the number of active qubits involved. This underlying quadratic scaling propagates throughout the entire evaluation pipeline. As a result, the maximum global cost to evaluate the full system of $m$ equations remains bounded by an $\mathcal{O}(q^2)$ complexity, driven by the actual number of intermediate monomials generated and the cumulative sum of the Hamming weights of their coefficients. This explains the non-intuitive reality that a single low-degree term with a massive, dense coefficient can legitimately demand more quantum resources than an entire system of multiple sparse, high-degree equations. Ultimately, compartmentalizing these parameters prevents multiplicative complexity explosions; because the maximum polynomial degree $d$ is typically far smaller than the total logical qubit count $q$ in practical applications, this overall $\mathcal{O}(q^2)$ spatial boundary dictates the true operational bottleneck of the architecture.

It is crucial to emphasize that this framework provides a generic, garbage-free quantum circuit for evaluating any polynomial over integer coefficients. The overall transformation manifests as the unitary operator:
\begin{equation} \label{eq:diof_operator}
        U_{f_j} \ket{\mathbf{x}}_{S}\ket{0}_{F}\ket{0}_{D_\ell} 
        = \ket{\mathbf{x}}_{S}\ket{f_j(\mathbf{x})}_{F}\ket{0}_{D_\ell}.
\end{equation}
This coherent polynomial evaluator is entirely agnostic to the broader quantum algorithm encompassing it. While utilized in this work as the oracle engine for Grover-based amplitude amplification, it operates as a universal arithmetic primitive. It is directly applicable to quantum walks over discrete graphs, Hamiltonian simulation of non-linear physical models, or any other quantum computational routine requiring coherent non-linear integer arithmetic.

\section{Oracle Implementation} \label{oracle_implementation}

Building upon the coherent arithmetic framework developed in Section~\ref{sec:quantum_representation}, we now detail the explicit construction of the complete Diophantine oracle required for Grover-based amplitude amplification. The role of this oracle is to coherently identify integer assignments $\mathbf{x}$ that satisfy all equations of the Diophantine system defined in Eq.~\eqref{eq:syst_dioph}, while strictly preserving reversibility to enable proper phase kickback onto the system register.

\subsection{Inequality Mapping and Sequential Logic}\label{oracle_implementation_A}

Given a candidate integer vector $\mathbf{x}$ encoded in the system register $S$, the oracle evaluates the Diophantine system and coherently accumulates the total number of satisfied conditions into a dedicated counter register $R$. To minimize the spatial circuit width, the oracle processes the constraints sequentially, systematically recycling a single main function register $F$.

However, checking standard equality constraints of the form $f_j(\mathbf{x}) = 0$ poses a severe hardware bottleneck. A naive implementation would compute the polynomial into $F$ and check if all its qubits are simultaneously in the $\ket{0}_F$ state. This necessitates a massive multi-controlled Toffoli gate whose depth scales linearly with the register width, destroying the efficiency of the oracle.

To drastically reduce this circuit complexity, we leverage the inequality mapping mechanism introduced in \cite{metropolis_ILP}. Based on Theorem 1 of \cite{metropolis_ILP}, any equality constraint $f_j(\mathbf{x}) = 0$ can be rigorously reformulated as a conjunction of two inequality constraints: $f_j(\mathbf{x}) \ge 0$ and $-f_j(\mathbf{x}) \ge 0$. As proven in the cited work, this mathematical reformulation guarantees a fundamental shift in the verification bottleneck: rather than requiring a macroscopic multi-controlled equality gate across the entire evaluation register, the problem is reduced to a simple binary sign-bit inspection, drastically curbing the overall gate complexity.

Consequently, the original system of $m$ equality constraints is expanded into a set of $2m$ inequality conditions. The counter register $R$ must therefore be sized to $\lceil \log_2(2m+1) \rceil$ qubits to represent states from $\ket{0}$ to $\ket{2m}$. The resulting global evaluation transformation takes the form:
\begin{equation}\label{rest_operator}
    \ket{\mathbf{x}}_{S}\ket{0}_{R}
    \;\longmapsto\;
    \begin{cases}
        \ket{\mathbf{x}}_{S}\ket{2m}_{R}, & \text{if } \mathbf{x} \in \Omega, \\[4pt]
        \ket{\mathbf{x}}_{S}\ket{r}_{R}, & \text{if } \mathbf{x} \notin \Omega,
    \end{cases}
\end{equation}
where $\Omega$ denotes the feasible set of complete integer solutions within the bounded domain, and $r < 2m$ is the number of inequality conditions satisfied by an infeasible point.

The profound advantage of this inequality mapping lies in its quantum implementation. In our two's-complement encoding, verifying a non-negativity condition trivially reduces to inspecting the sign bit (the MSB) of the function register. If the MSB is $\ket{0}$, the computed value is non-negative. Thus, assuming an arithmetic unitary $U_{f}$ computes a generic polynomial into $F$, we can construct a conditional update operator $A_{\text{ineq}}$ that acts concisely as:
\begin{equation}
    A_{\text{ineq}} \ket{\mathbf{x}}_{S} \ket{f (\mathbf{x})}_{F} \ket{r}_{R} := \ket{\mathbf{x}}_{S} \ket{f (\mathbf{x})}_{F} \ket{r + \delta_{\mathrm{MSB}, 0}}_{R},
\end{equation}
where $\delta_{\mathrm{MSB}, 0}$ evaluates to $1$ if the MSB is $\ket{0}$. 

Operationally, this test is realized by a single anti-controlled increment on the counter register $R$, targeted exclusively by the MSB of the function register. This completely eliminates the need for full-width multi-controlled operations, streamlining the equality check into an $\mathcal{O}(1)$ depth gate overhead relative to the function register size.

To evaluate the complete $j$-th equality constraint $f_j(\mathbf{x}) = 0$ without accumulating garbage entanglement, we compute, check, and uncompute both of its corresponding inequalities sequentially. We encapsulate this entirely into a single composite unitary $F_j$:
\begin{equation} \label{F_operator}
    F_j := \left( U_{-f_j}^{\dagger} A_{\text{ineq}} U_{-f_j} \right) \left( U_{f_j}^{\dagger} A_{\text{ineq}} U_{f_j} \right).
\end{equation}
Applying $F_j$ evaluates $f_j(\mathbf{x}) \ge 0$, increments the counter if true, uncomputes the register $F$, evaluates $-f_j(\mathbf{x}) \ge 0$, increments the counter if true, and uncomputes again. If $\mathbf{x}$ is a root of the equation, the counter is incremented exactly by $2$. If $f_j(\mathbf{x}) \neq 0$, it is incremented exactly by $1$ (since any non-zero integer is either strictly positive or strictly negative).

By sequentially applying the operators $F_j$ for $j=1,\ldots,m$, the overall counting operator $\prod_{j=1}^{m} F_j$ is assembled. The target state $\ket{2m}_R$ uniquely identifies valid global solutions that satisfy all equations. To complete the Grover oracle $O_f$, this specific state triggers a multi-controlled $Z$ gate (or an equivalent controlled phase flip) to apply a $-1$ phase exclusively to the solutions. Finally, the entire counting sequence is applied in reverse to uncompute the counter register back to $\ket{0}_R$, guaranteeing that the oracle acts unitarily and phase-kicks solely onto the system register. This complete, garbage-free algorithmic orchestration is depicted visually in the circuit schematic of Figure~\ref{fig:Grover_iteration}.

\subsection{Resource Estimation and Spatial Complexity} \label{subsec:complexity}

We now formalize the resource requirements derived throughout Sections \ref{sec:quantum_representation} and \ref{oracle_implementation} into a comprehensive bound encompassing both the spatial and non-Clifford gate complexities of the oracle.

\begin{theorem}[Space and Gate Complexity of the Bounded Diophantine Search Algorithm]\label{theor:complexity_total}
Consider a system of $m$ non-linear Diophantine polynomial equations in $n$ integer variables, of maximum degree $d \ge 2$, defined over a search domain $\mathcal{D} \subset \mathbb{Z}^n$ where each variable is restricted to an interval of length $N$. 
Let $\alpha = (\alpha_1, \dots, \alpha_n) \in \mathbb{N}^n$ be a tuple such that $\mathbf{x}^\alpha = x_1^{\alpha_1} \cdots x_n^{\alpha_n}$, and let $|\alpha| = \sum_{i=1}^n \alpha_i$ denote the total degree of the monomial. We express each polynomial as:
\begin{equation}
    f_j(\mathbf{x}) = \sum_{\alpha} c_{j,\alpha} \mathbf{x}^\alpha, \qquad |\alpha| \le d.
\end{equation}
Let $w_H(|c_{j,\alpha}|)$ denote the Hamming weight of the coefficient $c_{j,\alpha}$. Under the reversible compute-utilize-uncompute architecture, the complete Diophantine oracle $O_f$ can be implemented with a total number of logical qubits $q$ strictly bounded by:
\begin{equation}\label{eq:cost_total}
\begin{split}
    q &= n \lceil \log_2 N \rceil 
    + \max_{1 \le j \le m} \left\lceil \log_2 \left( \sum_{\alpha} |c_{j,\alpha}| N^{|\alpha|} \right) \right\rceil \\
    &\quad + \left( \sum_{\ell=2}^{d-1} \ell \lceil \log_2 N \rceil \right) 
    + \lceil \log_2(2m+1) \rceil + \mathcal{O}(1).
\end{split}
\end{equation}
In particular, the overall space complexity scales asymptotically as:
\begin{equation}\label{eq:cost_asint}
    q = \mathcal{O}\big((n + d^2)\log_2 N\big),
\end{equation}
where the bit-width contribution of the coefficient magnitudes acts strictly as a decoupled additive constant.

Furthermore, the non-Clifford gate complexity for a single complete Grover iteration satisfies:
\begin{equation}
    C_{\mathrm{Toff}}^{\mathrm{iter}} = \mathcal{O}(q^2 \Lambda),
\end{equation}
where $\Lambda$ captures the combined linear dependence on the complexity of the polynomial degrees and the Hamming weights of all coefficients across the system. A compatible explicit choice is:
\begin{equation}
    \Lambda \asymp \sum_{j,\alpha} \big(|\alpha| + w_H(|c_{j,\alpha}|)\big),
\end{equation}
where the symbol $\asymp$ denotes asymptotic equivalence in order of magnitude, i.e., $a \asymp b$ means $a = \Theta(b)$.

If the number of valid solutions is $M \ge 1$, then the total non-Clifford gate complexity for recovering one satisfying assignment without prior calibration scales as:
\begin{equation}
    C_{\mathrm{Toff}}^{\mathrm{total}} = \mathcal{O} \! \left( \sqrt{\frac{N^n}{M}} \, q^2 \Lambda \right).
\end{equation}
In the unique-solution case ($M=1$), this evaluates to the maximal search bound:
\begin{equation}
    C_{\mathrm{Toff}}^{\mathrm{total}} = \mathcal{O}\big(N^{n/2} \, q^2 \Lambda\big).
\end{equation}
\end{theorem}

\begin{proof}
We begin with the logical-qubit count.  
Each of the $n$ independent variables is encoded in binary arithmetic, requiring exactly $\lceil \log_2 N \rceil$ qubits to cover its interval of length $N$. Consequently, the system register $S$ contributes exactly $n\lceil \log_2 N \rceil$ qubits. 

The oracle evaluates the polynomial system sequentially by accumulating monomial contributions into a shared, bounded function register $F$. Since the absolute value of any variable within an interval of length $N$ can be conservatively upper-bounded by $N$ (regardless of its geometric offset from zero), the magnitude of any evaluated monomial satisfies:
\begin{equation}
    |c_{j,\alpha}\mathbf{x}^\alpha| \le |c_{j,\alpha}|\,N^{|\alpha|}.
\end{equation}
Because $F$ is uncomputed and dynamically recycled between equations, its maximum required capacity is dictated solely by the specific equation spanning the largest numerical range. Therefore, the function register size is strictly bounded by
\begin{equation}
    \max_{1 \le j \le m} \left\lceil \log_2\!\left(\sum_{\alpha}|c_{j,\alpha}|\,N^{|\alpha|}\right)\right\rceil
\end{equation}
qubits for the magnitude, plus a strictly constant $\mathcal{O}(1)$ overhead to accommodate sign bits, carry-ins, and temporary routing ancillas.

Next, consider the intermediate monomial registers used by the compute-utilize-uncompute architecture. For each intermediate degree $\ell \in \{2,\dots,d-1\}$, the construction provisions a dedicated register $D_\ell$ of width $\ell\lceil \log_2 N \rceil$ to prevent overflow during intermediate monomial generation. Therefore, the full intermediate workspace requires
\begin{equation}
    \sum_{\ell=2}^{d-1} \ell \lceil \log_2 N \rceil
\end{equation}
qubits. In addition, the inequality-checking stage utilizes a counter register of size $\lceil \log_2(2m+1)\rceil$ to encode the $2m$ inequalities. 

Summing all architectural components yields the exact global space requirement:
\begin{equation}
\begin{split}
    q \;&=\; n\lceil \log_2 N \rceil \;+\; \max_{1 \le j \le m} \left\lceil \log_2\!\left(\sum_{\alpha}|c_{j,\alpha}|\,N^{|\alpha|}\right)\right\rceil \\ 
    &\quad +\; \sum_{\ell=2}^{d-1} \ell \lceil \log_2 N \rceil \;+\; \lceil \log_2(2m+1)\rceil \;+\; \mathcal{O}(1).
\end{split}
\end{equation}
Asymptotically, the capacity of the function register contributes a spatial overhead strictly bounded by $\mathcal{O}\big(d \log_2 N + \max_j \log_2 ( \sum_{\alpha} |c_{j,\alpha}| )\big)$. Furthermore, the arithmetic progression of the intermediate workspace reduces to $\lceil \log_2 N \rceil \sum_{\ell=2}^{d-1} \ell = \mathcal{O}(d^2 \log_2 N)$. Because the $d \log_2 N$ evaluation cost is asymptotically absorbed by the larger $d^2 \log_2 N$ intermediate overhead, the global space complexity isolates the coefficient bit-length as a decoupled, additive factor. Consequently, the total logical qubit count scales cleanly as:
\begin{equation}
    q = \mathcal{O}\big((n+d^2)\log_2 N\big).
\end{equation}

We now bound the non-Clifford gate complexity of one complete Grover iteration. A Grover iteration consists of one application of the Diophantine oracle $O_f$ followed by the diffusion operator $D$. The diffusion step is a standard reflection over the system register, requiring an $\mathcal{O}(n \log_2 N)$ gate depth. This is asymptotically negligible compared to the arithmetic oracle, allowing us to absorb it and focus strictly on compiling $O_f$. 

Consider a single monomial term $c_{j,\alpha}\mathbf{x}^\alpha$ with total degree $|\alpha|$. Under the reversible architecture, the pure monomial is built through a sequence of pairwise controlled multiplications across bounded-width registers, followed by a coefficient-weighted accumulation into the main function register, and finally the inverse sequence that uncomputes all temporary data. Each bounded-width multiplication step is implemented through nested controlled shift-and-add operations. Because the arithmetic bottleneck for an in-place quantum addition scales linearly with the active register width, and every active register is strictly bounded by the total logical workspace $q$, the non-Clifford Toffoli depth of each pairwise multiplication is bounded by $\mathcal{O}(q^2)$. 

Moreover, evaluating a degree-$|\alpha|$ monomial requires $\mathcal{O}(|\alpha|)$ sequential stages. For $|\alpha| \ge 2$, the first $|\alpha|-1$ stages are intermediate multiplications between bounded variables, contributing a Toffoli depth of $\mathcal{O}(|\alpha| q^2)$. The final stage simultaneously multiplies the accumulated intermediate term by the final variable and the coefficient $c_{j,\alpha}$, accumulating the result into the function register. Because this final step is a nested controlled-addition modulated by the active bits of the coefficient, its Toffoli depth is exactly $\mathcal{O}(q^2 \, w_H(|c_{j,\alpha}|))$. Consequently, the local operational cost to evaluate and uncompute a single term $c_{j,\alpha}\mathbf{x}^\alpha$ naturally factors to:
\begin{equation}
    \mathcal{O}\!\big(|\alpha| q^2 + q^2 \, w_H(|c_{j,\alpha}|)\big) = \mathcal{O}\!\Big(q^2 \big(|\alpha| + w_H(|c_{j,\alpha}|)\big)\Big).
\end{equation}
Summing this exact complexity across all monomials in all $m$ equations yields the total oracle cost:
\begin{equation}
    C_{O_f} = \mathcal{O}\!\left( q^2 \sum_{j=1}^m \sum_{\alpha}\big(|\alpha| + w_H(|c_{j,\alpha}|)\big) \right).
\end{equation}
Therefore, by defining the composite arithmetic penalty parameter as
\begin{equation}
    \Lambda \asymp \sum_{j, \alpha}\big(|\alpha| + w_H(|c_{j,\alpha}|)\big),
\end{equation}
the total non-Clifford Toffoli cost of one Grover iteration satisfies
\begin{equation}
    C_{\mathrm{Toff}}^{\mathrm{iter}} = \mathcal{O}(q^2\,\Lambda).
\end{equation}
This formulation makes explicit that the dependence on the number of equations is purely additive, as the oracle processes the constraints sequentially, and that the coefficient overhead enters logarithmically via their Hamming weights. 

Finally, let $M\ge 1$ be the number of valid assignments within the bounded search space. Since the search space has a total cardinality of $N^n$, Grover amplitude amplification recovers one valid solution using $\mathcal{O}\!\big(\sqrt{N^n/M}\big)$ iterations. Multiplying this by the cost of one iteration gives the global algorithmic complexity:
\begin{equation}
    C_{\mathrm{Toff}}^{\mathrm{total}} = \mathcal{O}\!\left(\sqrt{\frac{N^n}{M}}\;q^2\,\Lambda\right).
\end{equation}
In the unique-solution case ($M=1$), this evaluates to the maximal search bound:
\begin{equation}
    C_{\mathrm{Toff}}^{\mathrm{total}} = \mathcal{O}\!\big(N^{n/2}q^2\,\Lambda\big).
\end{equation}
This completes the proof. 
\end{proof}

While the architecture is explicitly formulated for the symmetric two's complement interval $[-N/2, N/2 - 1]$, the derived asymptotic bounds are universally applicable to any contiguous search domain of length $N$, such as the unsigned interval $[0, N-1]$. Shifting the search space requires only an affine translation of the variables. Because our space complexity derivation conservatively bounds the maximum monomial magnitudes using the full interval width $N^{|\alpha|}$, both the overall logical qubit count $q$ and the associated Toffoli depth $\mathcal{O}(q^2 \Lambda)$ remain strictly invariant. The fundamental complexity is dictated entirely by the logarithmic width of the interval, not its geometric offset from zero.

This rigorous theoretical bound physically circumvents the massive spatial overhead typical of generic polynomial compilations, and is perfectly corroborated by the exact empirical resource scaling previously introduced in the left and right panel of Figure~\ref{fig:resource_estimation_combined}.

\subsection{Search Strategies for an Unknown Number of Solutions}
\label{sec:quantum_counting}

In practical applications, the number of valid solutions $M$ residing within the bounded search space of size $|\mathcal{D}| = N^n$ is typically unknown \textit{a priori}. The central difficulty in this regime is the calibration of the Grover iteration count: applying too few iterations under-amplifies the marked subspace, whereas applying too many leads to the well-known overshooting phenomenon, where the probability amplitude rotates past the target state and diminishes. To resolve this without sacrificing the quantum speedup, we tailor our approach to address three distinct operational objectives: finding a single satisfying assignment, iteratively enumerating all solutions, or determining their exact global count $M$.

\subsubsection{Iterative Retrieval and Solution Enumeration}

When the objective is to find a single feasible solution without explicitly counting the total volume first, one must employ quantum search procedures designed for an unknown number of marked states. Specifically, we adopt the BBHT search strategy~\cite{Boyer1998}. Although $M$ is unknown, this algorithm rigidly guarantees finding a solution by executing the Grover operator $G$ with a geometrically increasing, randomized sequence of iteration counts. The procedure operates as follows:

\begin{enumerate}[label=\arabic*), leftmargin=2.2em]
  \item[\bf (1)] Initialize a dynamic upper bound for the iterations $T=1$, and set a constant geometric growth factor $\lambda = 6/5$. This factor dictates the expansion rate of the search window; as established in Ref.~\cite{Boyer1998}, choosing any $1 < \lambda < 4/3$ safely prevents the schedule from skipping the optimal iteration regime.
  
  \item[\bf (2)] Choose an iteration count $j$ uniformly at random from the integer range $[0, T)$.
  
  \item[\bf (3)] Apply $j$ iterations of the operator $G$ to the initial uniform superposition and measure the system register $S$.
  
  \item[\bf (4)] If the measured state $\mathbf{x}^*$ is a feasible solution, output $\mathbf{x}^*$ and exit. (Note that while exhaustively searching the vast bounded domain $|\mathcal{D}|$ is classically intractable, evaluating the polynomial constraints to verify a measured candidate solution takes polynomial time, perfectly reflecting the verification asymmetry inherent to NP problems).
  
  \item[\bf (5)] Otherwise, update the boundary as $T \leftarrow \min(\lambda T, \sqrt{|\mathcal{D}|})$ and return to Step 2.
\end{enumerate}
As formally proven in Theorem 3 of Ref.~\cite{Boyer1998}, this geometric schedule guarantees finding a solution in an expected $\mathcal{O}(\sqrt{|\mathcal{D}|/M})$ queries without any prior calibration.

This iterative variant is especially powerful when the objective is to enumerate solutions rather than merely certify their existence. Upon measuring a valid assignment $\mathbf{x}^*$, we adopt a dynamic recovery strategy: the oracle is augmented to a new state $O_f'$ to explicitly exclude that specific assignment from subsequent searches. In our architecture, this is achieved by appending a multi-controlled X (MCX) gate---triggered exclusively when the system register $S$ is in state $\ket{\mathbf{x}^*}$---targeting the counter register $R$. By unitarily shifting the counter for this specific state, $\ket{\mathbf{x}^*}$ will tally fewer than the required $m$ satisfied equations, ensuring it is no longer marked by the phase flip.

Because the bounded search space contains $M$ feasible solutions, recovering all of them one by one requires a sequence of BBHT searches with exactly $M, M-1, \dots, 1$ marked states. The resulting total query complexity of this enumeration strategy scales as:
\begin{equation}\label{eq:iterative_sum}
    T_{\text{total}} = \sum_{r=1}^{M} \mathcal{O}\left(\sqrt{\frac{|\mathcal{D}|}{r}}\right) = \mathcal{O}\big(\sqrt{|\mathcal{D}| M}\big).
\end{equation}
This bound strictly preserves the quadratic quantum advantage over classical exhaustive enumeration, $\mathcal{O}(|\mathcal{D}|)$, whenever $M \ll |\mathcal{D}|$.

\subsubsection{Quantum Counting}

Alternatively, when the primary goal is to determine the exact number of solutions, the operator $G$ is integrated into a Quantum Counting framework~\cite{Brassard1998}. This framework interprets $G$ as a unitary transformation whose eigenvalues $e^{\pm i 2\theta}$ encode the ratio of marked states via $\sin^2(\theta) = M/|\mathcal{D}|$. By applying the Quantum Phase Estimation (QPE) primitive~\cite{Kitaev1995}---which utilizes an inverse Quantum Fourier Transform (QFT) to map the eigenphase $\theta$ onto an auxiliary evaluation register---one can coherently infer the exact value of $M$. The query complexity of this exact counting procedure fundamentally scales as $\mathcal{O}(\sqrt{|\mathcal{D}|})$ calls to the oracle $O_f$.

However, if the goal is merely to calibrate a subsequent deterministic search, an exact count is unnecessary; a multiplicative approximation of $M$ is sufficient to prevent overshooting. While the original Quantum Counting algorithm~\cite{Brassard1998} achieves this approximate counting to within a relative error $\epsilon$ in an optimal $\mathcal{O}\left(\frac{1}{\epsilon}\sqrt{\frac{|\mathcal{D}|}{M}}\right)$ queries, it fundamentally relies on the QFT, imposing a substantial spatial overhead due to the massive auxiliary evaluation register required.

To rigorously preserve the strict qubit-efficiency of our architecture, one may instead employ modern QFT-free approximate counting algorithms. For instance, the simplified protocol analyzed by Aaronson and Rall~\cite{Aaronson2020} achieves the identical optimal query complexity of $\mathcal{O}\left(\frac{1}{\epsilon}\sqrt{\frac{|\mathcal{D}|}{M}}\right)$ using exclusively variable-length sequences of standard Grover iterations alongside classical post-processing. By circumventing the phase estimation circuitry entirely, this contemporary paradigm completely eliminates the spatial overhead of the QFT, perfectly aligning with our memory-efficient design philosophy.

Regardless of the approximate estimation subroutine employed, the resulting estimate $\tilde{M}$ can subsequently be used to calibrate a deterministic retrieval run. Because the geometric rotation per Grover iteration is exactly $2\theta$, the required depth to reach the target subspace is given by:
\begin{equation}\label{eq:exact_grover_depth}
    L = \left\lfloor \frac{\pi}{4\arcsin(\sqrt{\tilde{M}/|\mathcal{D}|})} \right\rfloor,
\end{equation}
where the widely used $\frac{\pi}{4}\sqrt{|\mathcal{D}|/\tilde{M}}$ acts merely as its small-angle approximation. This exact calibration enables deterministic solution retrieval in a final, perfectly tuned search run.

\subsection{Classical vs. Quantum Complexity Separation}
\label{subsec:classical_quantum_comparison}

A persistent vulnerability in theoretical proposals for quantum search is the assumption of a ``black-box'' oracle with a constant $\mathcal{O}(1)$ execution cost. In practice, synthesizing the verification logic into a fully reversible, fault-tolerant quantum circuit can introduce massive spatial and temporal overheads. If the depth of the explicitly compiled oracle grows exponentially with the problem parameters, the theoretical quantum speedup is entirely nullified. The primary significance of our explicit architectural framework is the rigorous demonstration that this compilation overhead is strictly bounded.

To formulate a rigorous head-to-head comparison, we must evaluate both classical and quantum costs using a unified metric of fundamental logical operations. In classical boolean arithmetic, the asymptotic algorithmic bottleneck is governed by its multiplicative complexity---the number of non-linear logic gates (e.g., AND gates) required to evaluate the polynomial system~\cite{Brgisser1997, 10.5555/35517}. Quantumly, within a fault-tolerant regime, the non-Clifford Toffoli count serves as the direct mathematical equivalent for reversible logic, as Toffoli gates inherently implement logical conjunctions while Clifford operations remain asymptotically subdominant~\cite{Wang2025}. By quantifying both regimes via their non-linear operational bottlenecks, and building upon the exact architectural bounds established in Theorem~\ref{theor:complexity_total}, we can state the following explicit separation.

First, we establish the classical algorithmic baseline by formalizing the multiplicative complexity required to evaluate the Diophantine system. It is important to emphasize that while specific, highly structured Diophantine systems may admit specialized classical heuristics, the general bounded problem is NP-complete. Consequently, any exact algorithm that guarantees correctness for arbitrary non-linear systems must fundamentally rely on exhaustive enumeration in the worst case.

\begin{lemma}[Worst-Case Classical Multiplicative Complexity of Diophantine Verification]\label{lemma:class_complex}
Let $\mathcal{D} \subset \mathbb{Z}^n$ be a search domain where each of the $n$ variables is restricted to an interval of length $N$. Consider a system of $m$ Diophantine polynomial equations $\{f_j(\mathbf{x}) = 0\}_{j=1}^m$, where $f_j(\mathbf{x}) = \sum_{\alpha} c_{j,\alpha} \mathbf{x}^\alpha$ with maximum degree $d \geq 2$. 

Let $q_{\mathrm{cl}}$ denote the maximum classical bit-width required to evaluate any intermediate or final polynomial value without overflow, and let $\Lambda \asymp \sum_{j,\alpha} \big(|\alpha| + w_H(|c_{j,\alpha}|)\big)$ denote the cumulative arithmetic penalty dictated by the polynomial degrees and the coefficient Hamming weights. The classical multiplicative bit-complexity (the minimal number of non-linear boolean gates, e.g., AND gates) required to verify whether a single candidate assignment $\mathbf{x} \in \mathcal{D}$ is a valid solution is strictly bounded by:
\begin{equation}
    C_{\mathrm{eval}} = \mathcal{O}(q_{\mathrm{cl}}^2 \Lambda)
\end{equation}
\end{lemma}

\begin{proof}
We first determine the maximum classical bit-width $q_{\mathrm{cl}}$ required to evaluate the system. Classically, the independent variables initially require $\lceil \log_2 N \rceil$ bits to encode their intervals. During the sequential evaluation of the polynomials, the absolute value of any evaluated monomial is bounded by $|c_{j,\alpha}| N^{|\alpha|}$. Consequently, the maximum possible numerical value attained during the evaluation of any equation $f_j$ requires at most $\lceil \log_2 ( \sum_{\alpha} |c_{j,\alpha}| N^{|\alpha|} ) \rceil$ bits. Accounting for intermediate arithmetic bounds, the classical memory requirement $q_{\mathrm{cl}}$ scales identically to the architectural quantum logical qubit bound $q$ derived in Theorem~\ref{theor:complexity_total}. Because both quantities are governed by the same monomial-magnitude upper bounds, they admit the exact same asymptotic expression $q_{\mathrm{cl}} \asymp q = \mathcal{O}((n + d^2) \log_2 N)$. Thus, any pairwise classical arithmetic operation acts on a bit-string of length at most $q_{\mathrm{cl}}$.

In boolean circuit theory, computing the carries when adding two $w$-bit integers requires $\mathcal{O}(w)$ non-linear logic gates (AND gates). Multiplying two $w$-bit integers requires $\mathcal{O}(w^2)$ non-linear gates. 

To evaluate a single monomial $c_{j,\alpha}\mathbf{x}^\alpha$, the classical algorithm must compute the variable product $\mathbf{x}^\alpha$. For a term of degree $|\alpha| \ge 2$, building the intermediate pure monomial requires $|\alpha|-1$ sequential multiplications, costing $\mathcal{O}(|\alpha| q_{\mathrm{cl}}^2)$ AND gates. 

Crucially, the final arithmetic stage multiplies this intermediate variable state by the integer coefficient $c_{j,\alpha}$. Following a resource-efficient nested shift-and-add logic equivalent to the reversible architecture, this final coefficient-weighted multiplication requires $\mathcal{O}(q_{\mathrm{cl}})$ shifted additions for each of the $w_H(|c_{j,\alpha}|)$ active bits in the coefficient's binary expansion. Because each addition operates on $q_{\mathrm{cl}}$ bits, this nested accumulation inherently costs $\mathcal{O}(q_{\mathrm{cl}}^2 \cdot w_H(|c_{j,\alpha}|))$ non-linear gates.

Summing these non-linear gate contributions over all monomials across the $m$ equations yields the exact worst-case classical evaluation cost per candidate assignment:
\begin{equation}
    C_{\mathrm{eval}} = \mathcal{O}\left( \sum_{j=1}^m \sum_{\alpha} \Big( |\alpha| q_{\mathrm{cl}}^2 + q_{\mathrm{cl}}^2 \, w_H(|c_{j,\alpha}|) \Big) \right)
\end{equation}
By factoring out $q_{\mathrm{cl}}^2$, we establish the rigorous upper bound:
\begin{equation}
    C_{\mathrm{eval}} = \mathcal{O}\left( q_{\mathrm{cl}}^2 \sum_{j,\alpha} \big(|\alpha| + w_H(|c_{j,\alpha}|)\big) \right) = \mathcal{O}(q_{\mathrm{cl}}^2 \Lambda)
\end{equation}
Because $q_{\mathrm{cl}} \asymp q$, this result explicitly demonstrates that the worst-case classical non-linear boolean complexity is structurally governed by the exact same fundamental arithmetic bounds ($q$ and $\Lambda$) as our memory-efficient reversible quantum architecture, satisfying $C_{\mathrm{eval}} = \mathcal{O}(q^2 \Lambda)$.
\end{proof}

As explicitly demonstrated by Lemma~\ref{lemma:class_complex} and Theorem~\ref{theor:complexity_total}, the fundamental arithmetic bottleneck for evaluating a candidate assignment is asymptotically identical in both regimes ($\mathcal{O}(q^2 \Lambda)$). Consequently, compiling the verification logic into a reversible quantum circuit yields no intrinsic algebraic advantage over classical boolean evaluation. The quantum computational advantage emerges entirely from the algorithmic framework in which this arithmetic is deployed---specifically, the capacity to coherently query the evaluation logic across a global superposition. With both the classical baseline and the quantum arithmetic overhead rigorously bounded and structurally matched, we can formally derive the net algorithmic separation.

\begin{theorem}[Classical--Quantum Complexity Separation]\label{theor:complexity_separation}
Consider a bounded polynomial Diophantine system defined over the search domain $\mathcal{D} \subset \mathbb{Z}^n$ containing an unknown number of $M \ge 1$ valid solutions. Let $T_{\mathrm{cl}}$ denote the worst-case classical boolean complexity to exhaustively search $\mathcal{D}$ and guarantee the enumeration of all solutions. Let $T_{\mathrm{q}}$ denote the non-Clifford Toffoli complexity of the explicit quantum framework to dynamically isolate the complete set of $M$ assignments. The asymptotic complexity separation ratio between the classical and quantum operational costs is strictly governed by the algorithmic speedup:
\begin{equation}\label{eq_comparison_cl_q}
    \frac{T_{\mathrm{cl}}}{T_{\mathrm{q}}} = \Omega\left(\sqrt{\frac{N^n}{M}}\right) 
\end{equation}
\end{theorem}

\begin{proof}
To guarantee the complete enumeration of all $M$ valid assignments without prior knowledge of their total count or distribution, a worst-case classical algorithm must sequentially inspect the entire finite domain of cardinality $|\mathcal{D}| = N^n$. By Lemma~\ref{lemma:class_complex}, the non-linear operational cost per evaluation is governed by $C_{\mathrm{eval}} = \mathcal{O}(q^2 \Lambda)$. Thus, the total classical algorithmic complexity for full enumeration is:
\begin{equation}
    T_{\mathrm{cl}} = \Theta\Big(N^n \cdot C_{\mathrm{eval}}\Big) 
\end{equation}
In the quantum regime, our reversible architecture completely replaces the exhaustive spatial search with coherent amplitude amplification. As established in Section~\ref{sec:quantum_counting}, extracting the complete set of $M$ valid assignments dynamically via the iterative BBHT schedule requires an accumulated $\mathcal{O}(\sqrt{N^n \cdot M})$ queries. Crucially, as derived in Theorem~\ref{theor:complexity_total}, the maximum fault-tolerant Toffoli depth of evaluating the system introduces an arithmetic overhead of exactly $\mathcal{O}(q^2 \Lambda)$. Therefore, the total quantum algorithmic complexity for complete enumeration is:
\begin{equation}
    T_{\mathrm{q}} = \mathcal{O}\Big(\sqrt{N^n \cdot M} \cdot q^2 \Lambda\Big) 
\end{equation}
Taking the ratio of the total classical and quantum operational costs yields:
\begin{equation}
    \frac{T_{\mathrm{cl}}}{T_{\mathrm{q}}} = \frac{\Theta\big(N^n \cdot C_{\mathrm{eval}}\big)}{\mathcal{O}\big(\sqrt{N^n \cdot M} \cdot q^2 \Lambda\big)} = \Omega\left( \sqrt{\frac{N^n}{M}} \cdot \frac{C_{\mathrm{eval}}}{q^2 \Lambda} \right) 
\end{equation}
Because Lemma~\ref{lemma:class_complex} and Theorem~\ref{theor:complexity_total} establish that both architectures share the exact same asymptotic arithmetic bottleneck $\mathcal{O}(q^2 \Lambda)$---derived from the identical underlying nested shift-and-add logic---their structural quotient $\frac{C_{\mathrm{eval}}}{q^2 \Lambda}$ is bounded below by a constant $\Omega(1)$. As a result, the arithmetic algorithmic overheads effectively mirror and absorb one another. Consequently, the true lower bound of the advantage is strictly defined by the search space reduction, proving the separation.
\end{proof}

It is worth noting that this identical separation ratio holds even if the operational objective is relaxed. If the task is strictly to find a \textit{single} valid assignment, a classical algorithm randomly inspecting candidates will succeed in an expected $\mathcal{O}(N^n/M)$ evaluations. Since the corresponding BBHT quantum search requires an expected $\mathcal{O}(\sqrt{N^n/M})$ queries to isolate one solution, the asymptotic ratio between the classical and quantum costs remains precisely $\Omega(\sqrt{N^n/M})$.

The relation in Eq.~\eqref{eq_comparison_cl_q} encapsulates the core of the quantum computational advantage for these mathematical problems. Because the geometric size of the search space fundamentally dwarfs the structural parameters of the polynomial system ($N^n \gg q^2 \Lambda$), the arithmetic overhead acts merely as a subdominant prefactor. Consequently, while the underlying Diophantine problem remains asymptotically hard---requiring $\mathcal{O}(N^{n/2})$ iterations in the worst case---the architecture achieves a genuine quadratic speedup over classical exhaustive enumeration. Furthermore, this advantage is analytically robust against potential classical optimizations. Even if specialized heuristics reduce the classical evaluation cost $C_{\mathrm{eval}}$ below our worst-case bound, exact verification inherently requires at least polynomial time. Since the quantum oracle overhead is also strictly polynomial, no classical boolean optimization can asymptotically absorb the diverging search-space factor, ensuring that the foundational quantum speedup is structurally preserved.

\section{Simulation Results} \label{sec:results}

To rigorously validate the theoretical complexity bounds and evaluate the true scalability of the proposed quantum architecture, we performed comprehensive, implementation-level resource estimations and coherent state-vector simulations. The complete algorithmic framework was synthesized into explicit unitary circuits using \textit{Qiskit}~\cite{qiskit2024}. Rather than relying on heuristic asymptotic estimates or bounding solely the arithmetic oracle $O_f$, the full Grover search operator was explicitly constructed. Its logical gate composition was then exactly counted across thousands of randomly generated Diophantine instances, encompassing all subroutines—including the inequality mapping logic, the uncomputation steps, and the standard diffusion operator. This rigorous compilation yields a hardware-agnostic, deterministic measure of logical complexity that reflects the true computational overhead of the coherent search.

\subsection{Resource Metrics and Gate Cost Model}

To ensure a fair and architecture-independent comparison, the computational overhead is quantified primarily at the level of Toffoli-equivalent gates. This metric is motivated by the physical constraints of fault-tolerant quantum computing: logical Clifford operations (such as CNOT, $H$, $S$, and $X$) can typically be implemented transversally in topological color codes~\cite{CC1, CC2}, for which key fault-tolerant ingredients have now been experimentally explored across trapped-ion, superconducting, and neutral-atom platforms~\cite{ION1, GOOGLE1, QUERA1}. Conversely, non-Clifford operations strictly dominate the overall execution bottleneck due to the rigorous hardware demands of magic-state distillation~\cite{Wang2025, Ortega2025}. Therefore, expressing the total logical gate count in Toffoli-equivalents provides the most meaningful indicator of the actual scaling limitations.

\begin{figure*}[t] 
  \centering 
  \begin{subfigure}{0.49\textwidth} 
    \centering 
    \includegraphics[width=\linewidth]{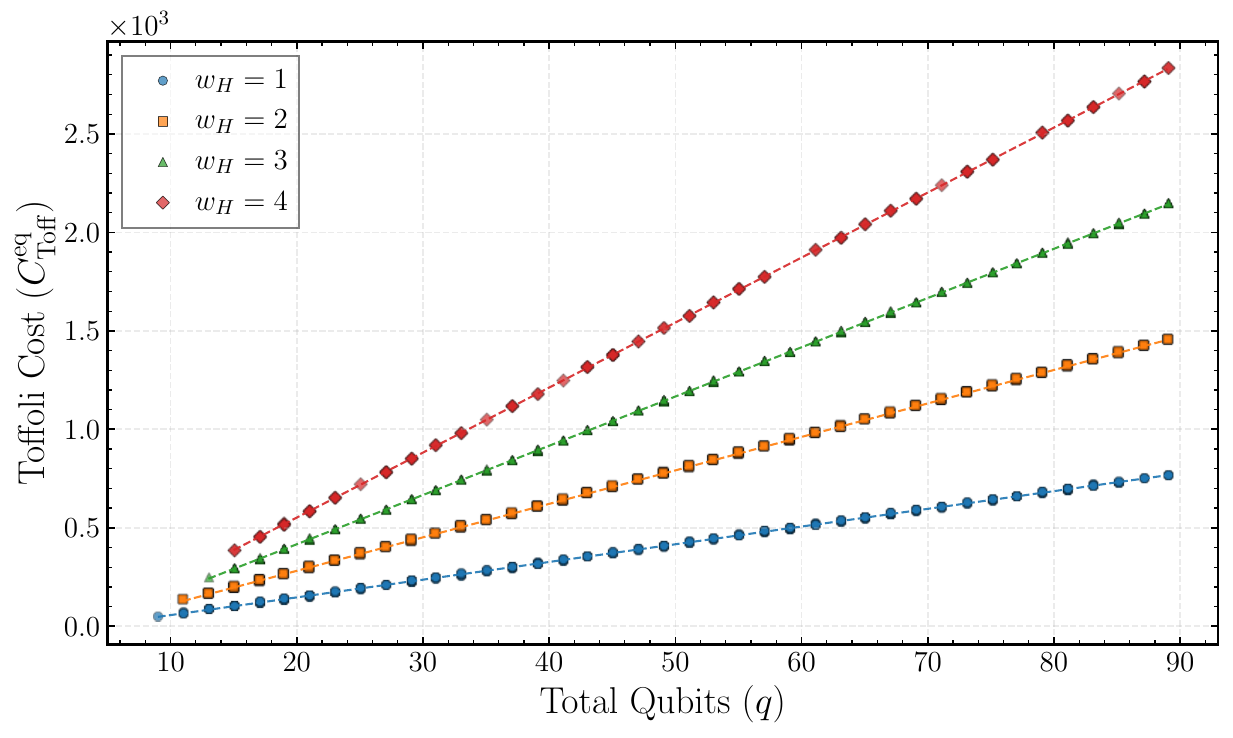} 
  \end{subfigure}\hfill 
  \begin{subfigure}{0.49\textwidth} 
    \centering 
    \includegraphics[width=\linewidth]{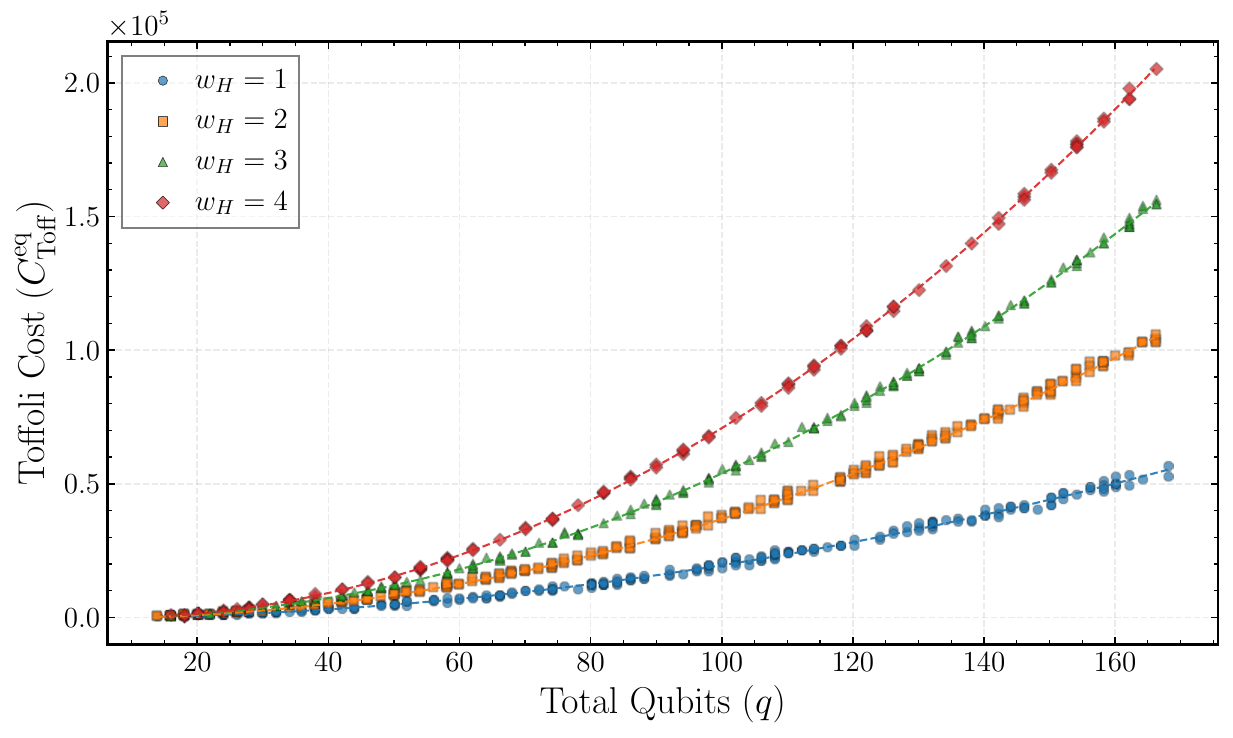} 
  \end{subfigure}
  \caption{\justifying 
\textbf{Empirical complexity scaling of the Grover operator for linear and quadratic polynomials.} 
Resource estimation analysis for Diophantine oracles with a single variable ($n=1$) across 1500 problem instances. 
\textbf{(Left)} Toffoli gate count ($C_{\mathrm{Toff}}^{\mathrm{eq}}$) versus total number of logical qubits ($q$) for linear Diophantine equations of the form $c_1 x + c_0 = 0$, exhibiting strictly linear complexity growth ($\mathcal{O}(q)$) where the slope of this linear scaling is directly proportional to the Hamming weight of the leading coefficient, $w_H(|c_1|)$. 
\textbf{(Right)} Toffoli gate count ($C_{\mathrm{Toff}}^{\mathrm{eq}}$) versus total number of logical qubits ($q$) for quadratic Diophantine equations of the form $c_2 x^2 + c_1 x + c_0 = 0$, illustrating the strict $\mathcal{O}(q^2)$ envelope associated with higher-degree nested arithmetic. In both panels, the gate complexity scales proportionally with the Hamming weight of the respective leading coefficient. Because this weight is inherently bounded by the integer bit-length, the empirical behavior formally validates an overall non-Clifford upper bound of $\mathcal{O}(q \log_2 |c_1|)$ and $\mathcal{O}(q^2 \log_2 |c_2|)$ for the linear and quadratic cases, respectively.
  }
  \label{fig:toffoli_lin_quad}
\end{figure*}

In this work, we adopt a unified Toffoli-equivalent cost model~\cite{Nielsen_Chuang_2010, Ortega2025}. Any natively synthesized standard Toffoli (CCX) gate structurally corresponds to exactly one Toffoli-equivalent. For multi-controlled $X$ gates (MCX) with $n_c$ control qubits, the cost contribution is assigned according to the well-established linear-depth decomposition:
\begin{equation}
    C_{\mathrm{MCX}}(n_c) =
\begin{cases}
0, & n_c < 2, \\
1, & n_c = 2, \\
2n_c - 3, & n_c > 2.
\end{cases}
\end{equation}

These cases correspond, respectively, to zero-cost Clifford operations (such as standard CNOTs or unconditioned $X$ gates), a standard Toffoli gate ($n_c=2$), and a linear-depth cascaded decomposition. This higher-order decomposition ($n_c > 2$) systematically utilizes intermediate ancilla qubits to break down the macroscopic multi-control condition into a sequential chain of exactly $2n_c - 3$ standard Toffoli gates~\cite{Ortega2025}.

Furthermore, the synthesis of complex arithmetic sequences occasionally leaves behind independent, single-qubit non-Clifford rotations ($T$ and $T^\dagger$ gates). To unify these isolated rotations into our primary metric, their cumulative contribution is converted into Toffoli-equivalents based on standard magic-state distillation ratios. Because a single logical Toffoli gate typically requires seven $T$-gate equivalents to be fault-tolerantly synthesized, we define the standalone rotation cost as:
\begin{equation}
C_T = \frac{N_T + N_{T^\dagger}}{7},
\end{equation}
where $N_T$ and $N_{T^\dagger}$ denote the total counts of individual $T$ and $T^\dagger$ gates in the transpiled circuit. All baseline Clifford gates are assigned a negligible cost weight of zero. 

Consequently, the total non-Clifford depth of a complete Grover iteration instance is formally evaluated by the sum:
\begin{equation}
C_{\mathrm{Toff}}^{\mathrm{eq}} = \sum_{i=1}^{N_{\mathrm{MCX}}} C_{\mathrm{MCX}}(n_{c, i}) + C_T + N_{\mathrm{CCX}},
\end{equation}
where $N_{\mathrm{MCX}}$ is the total number of strictly multi-controlled $X$ gates ($n_c > 2$), $n_{c, i}$ denotes the specific number of control qubits for the $i$-th such gate, and $N_{\mathrm{CCX}}$ is the total count of standard Toffoli gates. By applying this exact cost model within the Qiskit framework, we transparently map the raw topological circuit depth to its true fault-tolerant operational cost.

\begin{figure*}[t] 
  \centering 
  \begin{subfigure}{0.49\textwidth} 
    \centering 
    \includegraphics[width=\linewidth]{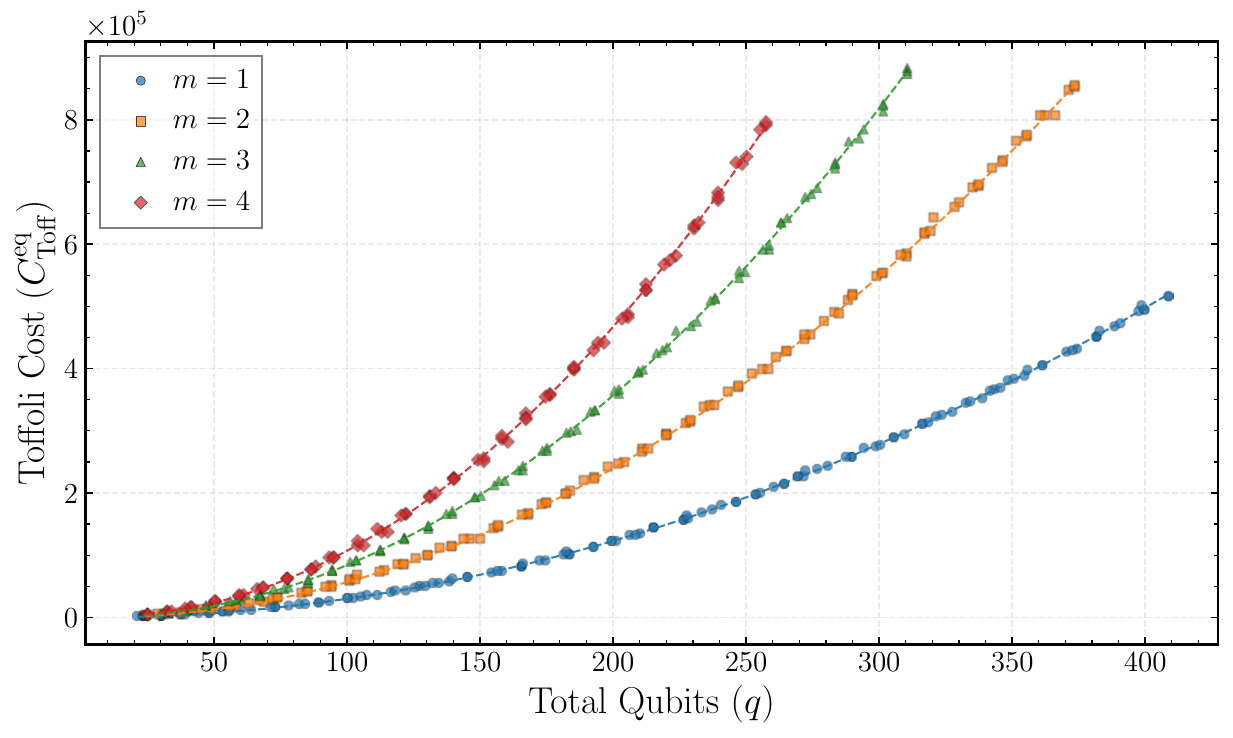} 
  \end{subfigure}\hfill 
  \begin{subfigure}{0.49\textwidth} 
    \centering 
    \includegraphics[width=\linewidth]{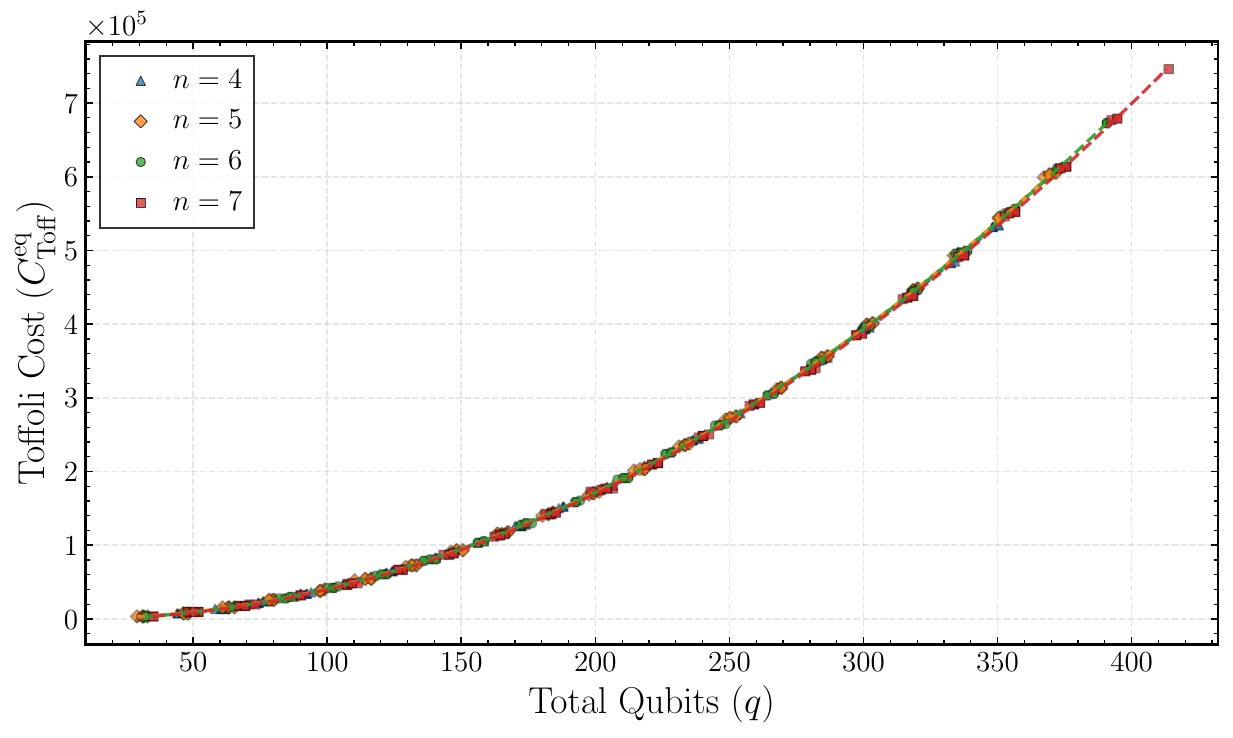} 
  \end{subfigure}

  \caption{\justifying 
    \textbf{Scaling analysis of the Diophantine oracle isolating the impact of system dimensionality.} 
    \textbf{(Left)} Toffoli gate count ($C_{\mathrm{Toff}}^{\mathrm{eq}}$) versus total number of logical qubits ($q$) for systems with varying numbers of equality constraints ($m \in [1, 4]$), evaluated at a fixed maximum polynomial degree ($d=3$) and a constant cumulative Hamming weight per equation. Adding equations acts as a linear scaling factor over the established quadratic baseline, demonstrating that the gate complexity grows strictly proportionally to $m$ without altering the fundamental scaling envelope. 
    \textbf{(Right)} Toffoli gate count ($C_{\mathrm{Toff}}^{\mathrm{eq}}$) versus total number of logical qubits ($q$) for systems with varying numbers of variables ($n \in [4, 7]$), evaluated at a fixed degree ($d=4$) and a fixed cumulative coefficient weight. Remarkably, increasing the independent variables leaves the fundamental $\mathcal{O}(q^2)$ scaling exponent completely invariant. Because the total logical qubit count $q$ inherently encompasses the state space for all $n$ variables, the computational burden of a wider variable set is natively absorbed into the horizontal $q$-axis, validating the robustness of sequentially recycling the central accumulator.
  }
  \label{fig:toffoli_eqs_vars}
\end{figure*}

\subsection{Empirical Complexity of the Grover Operator}

We first analyze the arithmetic core by isolating linear and quadratic Diophantine equations for a single variable ($n=1$). Figure~\ref{fig:toffoli_lin_quad} presents the Toffoli gate count ($C_{\mathrm{Toff}}^{\mathrm{eq}}$) against the total number of logical qubits ($q$) across 1500 randomly generated instances. The left panel aligns with the theoretical predictions of Section~\ref{subsec:linear_arithmetic}: purely linear equations ($c_1 x + c_0 = 0$) exhibit a linear gate depth ($\mathcal{O}(q)$). The slope of this linear scaling is explicitly determined by the Hamming weight of the leading coefficient, $w_H(|c_1|)$. By isolating this variable, the empirical data confirms that each active bit in the coefficient's binary expansion imposes a discrete, constant additive overhead. Because this Hamming weight is upper-bounded by the integer's bit-length, the algorithmic depth scales logarithmically with the absolute magnitude of the coefficient, validating the $\mathcal{O}(q \log_2 |c_1|)$ boundary.

The right panel of Figure~\ref{fig:toffoli_lin_quad} illustrates the transition to quadratic Diophantine equations ($c_2 x^2 + c_1 x + c_0 = 0$). Here, the necessity of nested shift-and-add multiplications shifts the computational complexity into an $\mathcal{O}(q^2)$ regime. Stratifying the simulated data by the Hamming weight of the leading non-linear coefficient, $w_H(|c_2|)$, reveals a structured resource profile: the leading scaling constant is directly proportional to the number of non-zero bits in $c_2$. This confirms that the maximum arithmetic depth is bounded by $\mathcal{O}(q^2 \log_2 |c_2|)$, as derived in Theorem~\ref{theor:complexity_total}. The observed intra-weight dispersion arises from the randomized linear terms ($c_1$). Nevertheless, this lower-order variance is asymptotically eclipsed by the predictable $\mathcal{O}(q^2)$ envelope, confirming that the high-degree in-place accumulation constitutes the true resource bottleneck.

To evaluate the robustness of the architecture against system dimensionality, Figure~\ref{fig:toffoli_eqs_vars} contrasts the impact of expanding the system of equations ($m$) versus increasing the number of variables ($n$). The left panel demonstrates that for a fixed maximum degree ($d=3$) and a constant cumulative Hamming weight per equation, increasing the number of equations ($m \in [1, 4]$) yields a purely additive vertical shift in the Toffoli depth. Because the oracle evaluates constraints sequentially using the same recycled accumulator $F$, adding equations is algorithmically indistinguishable from evaluating a single equation with a proportionally larger aggregate Hamming weight. Consequently, $m$ manifests strictly as a linear prefactor over the established $\mathcal{O}(q^2)$ baseline, rather than triggering a multidimensional complexity explosion. Conversely, the right panel isolates the effect of variable scaling ($n \in [4, 7]$) at a fixed degree ($d=4$) and a fixed cumulative coefficient weight. Remarkably, increasing the independent variables leaves the fundamental $\mathcal{O}(q^2)$ scaling exponent perfectly invariant. Because the total logical qubit count $q$ inherently encompasses the state space for all $n$ variables, the computational burden of a wider variable set is natively absorbed into the horizontal $q$-axis. By routing all arithmetic sequentially through the central framework, the architecture effectively immunizes the algorithm against variable-induced depth explosions.

The overall performance of the generalized Grover operator is synthesized in the comprehensive analysis of Figure~\ref{fig:resource_estimation_combined} (previously referenced in Section~\ref{sec:intro}), encompassing extensive systems with variables $n \in [1, 7]$ and arbitrary polynomial degrees $d \in [2, 7]$. A log-log regression of the Toffoli depth (left panel) yields an empirical scaling exponent of $1.77$. This formally confirms that, across the practical pre-asymptotic regime, the actual circuit depth exhibits a highly optimized sub-quadratic growth, operating comfortably within the worst-case $\mathcal{O}(q^2)$ theoretical boundary. This sub-quadratic efficiency empirically demonstrates the success of dynamic zero-skipping within the quantum-controlled routines. Crucially, this optimization extends seamlessly to the spatial domain. The empirical logical qubit overhead (right panel) demonstrates a strictly linear correlation with the theoretical problem width, perfectly validating the analytical space complexity bound of $q = \mathcal{O}((n + d^2)\log_2 N)$ derived in Eq.~\eqref{eq:cost_asint}. By rigorously executing the intermediate monomial uncomputation strategy, the framework definitively confines the spatial requirements to this compact formulation, proving that the theoretical memory limits translate flawlessly into scalable hardware synthesis.

\subsection{Validation of Coherent Amplitude Amplification}

Although resource scaling guarantees architectural efficiency, it does not inherently prove the correctness of the quantum phase manipulations. To validate the strict reversibility and operational accuracy of our formulation, we simulated the complete Grover search algorithm over a strongly coupled multivariate quadratic system. 

As detailed in Figure~\ref{fig:prob_vs_step}, we encoded an $n=3$ variable Diophantine system ($x,y,z$) discretized at $\lceil \log_2 N \rceil = 3$ qubits per variable (accounting for the two's complement sign bit). This defines a finite search domain of size $|\mathcal{D}| = N^n = 2^9 = 512$. To rigorously test the uncomputation routines against cross-register entanglement, the benchmark system was explicitly designed with heavily coupled non-linear terms ($xy, yz, xz$) and non-trivial bounds, possessing a unique valid integer assignment at $\ket{x=3, y=2, z=1}_S$:
\begin{equation}
    \begin{cases}
        3x^2+2y^2+5z^2 = 40, \\
        2xy-4yz+3xz = 13, \\
        -x^2+5y-7z = -6.
    \end{cases}
\end{equation}

The state-vector simulation results exhibit pristine amplitude amplification dynamics. The probability of measuring the target state grows monotonically, seamlessly navigating the complex algebraic landscape generated by the $O_f$ oracle. Crucially, the success probability peaks at $99.9\%$ exactly at iteration $t = 17$, which perfectly matches the theoretical optimum predicted by Grover's formulation for a single marked state: $t_{\mathrm{opt}} = \lfloor \frac{\pi}{4}\sqrt{|\mathcal{D}|} \rfloor = 17$. 

This deterministic convergence provides critical empirical proof for the structural integrity of the oracle. It confirms that the inequality-mapping protocol (via the MSB check) and the dynamic uncomputation routines successfully clear all temporary registers without leaving residual garbage entanglement or parasitic phase errors. Ultimately, demonstrating that such a dense, cross-coupled polynomial system can be coherently resolved and simulated within a compact logical workspace physically validates the theoretical resource efficiency of the proposed architecture.

\section{Conclusions} \label{sec:conclusions}

In this work, we have introduced a reversible quantum arithmetic framework for evaluating non-linear Diophantine equations. Moving decisively beyond abstract ``black-box'' assumptions, we provided an explicit, end-to-end gate-level construction capable of computing arbitrary integer polynomials. While evaluated here within the context of quantum search, this arithmetic engine is inherently algorithm-agnostic. It serves as a concrete computational primitive ready for integration into broader quantum protocols requiring coherent integer arithmetic, such as quantum walks or advanced optimization heuristics.

The architecture's efficiency stems from a set of targeted design choices: virtual rewiring for zero-cost scalar multiplication, a sequential compute-utilize-uncompute strategy for bounded monomial generation, and MSB-based inequality mapping for $\mathcal{O}(1)$-depth equality checks. The solidity of this framework relies on the correspondence between our analytical derivations and empirical evidence. Through implementation-level synthesis across thousands of randomized problem instances, we validated that these techniques bypass the prohibitive overheads typical of generic nested multipliers. Our Toffoli counts confirm that the non-Clifford depth per Grover iteration operates within the theoretical boundary of $\mathcal{O}(q^2 \Lambda)$, while the spatial complexity is confined to a logical workspace of $q = \mathcal{O}((n + d^2) \log_2 N)$ qubits. This alignment between theory and simulation ensures a self-consistent and robust resource estimation.

Crucially, this explicit architectural synthesis allows us to evaluate the classical-quantum complexity separation with rigorous transparency. Because the compiled arithmetic oracle requires strictly bounded polynomial overheads in space and time, it theoretically preserves the fundamental algorithmic speedup of amplitude amplification. While achieving an absolute wall-clock advantage in practice will depend heavily on overcoming the inherent clock-speed disparities between classical CPUs and fault-tolerant quantum hardware, our framework mathematically guarantees that the oracle compilation does not introduce hidden exponential bottlenecks. Consequently, whether identifying a unique solution or dynamically enumerating valid assignments via iterative geometric schedules, the architecture ensures that the theoretical $\mathcal{O}(\sqrt{|\mathcal{D}|})$ Grover scaling remains cleanly intact.

Finally, our Toffoli-equivalent resource estimation provides a deterministic blueprint for evaluating the feasibility of this architecture on future fault-tolerant hardware. Building upon this baseline, future research could adapt this framework to structured families of polynomial systems arising in cryptography or discrete optimization, exploiting specific algebraic symmetries to further compress the compilation depth. Ultimately, this work provides a scalable, self-consistent, and transparent toolkit for mapping complex non-linear integer problems directly into practically compilable quantum circuits.

\section{Acknowledgements}
G.E. and M.A.M.-D. acknowledge the support from grants MINECO/FEDER Projects, PID2021-122547NB-I00 FIS2021, MADQuantumCM project funded by Comunidad de Madrid, the Recovery, Transformation, and Resilience Plan, NextGenerationEU, funded by the European Union, and the Ministry of Economic Affairs Quantum ENIA project funded by Madrid ELLIS Unit CAM. G. E. also acknowledge the support from the CAM Program TEC-2024/COM-84 QUITEMAD-CM.
M.A.M.-D. has also been partially supported by the U.S. Army Research Office through Grant No.W911NF-14-1-0103. This work has been financially supported by the Ministry for Digital Transformation and of Civil Service of the Spanish Government through the QUANTUM ENIA project call – Quantum Spain project, and by the European Union through the Recovery, Transformation and Resilience Plan – NextGenerationEU within the framework of the Digital Spain 2026 Agenda.

\newpage

\bibliography{bibliography}

\end{document}